\PassOptionsToPackage{hidelinks}{hyperref}
\documentclass[referee]{sn-jnl}%
\AtBeginDocument{\urlstyle{same}}
\usepackage{graphicx}%
\usepackage{multirow}%
\usepackage{amsmath,amssymb,amsfonts}%
\usepackage{amsthm}%
\usepackage{mathrsfs}%
\usepackage[title]{appendix}%
\usepackage{xcolor}%
\usepackage{textcomp}%
\usepackage{manyfoot}%
\usepackage{booktabs}%
\usepackage{algorithm}%
\usepackage{algorithmicx}%
\usepackage{algpseudocode}%
\usepackage{listings}%
\theoremstyle{plain}
\newtheorem{thm}{Theorem}
\newtheorem{cor}[thm]{Corollary}

\newtheorem{pro}{Proposition}
\theoremstyle{definition}
\newtheorem{definition}{Definition}[section]
\newtheorem{remark}{Remark}[section]
\numberwithin{equation}{section}
\newcommand{\cstar}{C^{\ast}}%
\newcommand{\R}{{\mathbb{R}}}%
\newcommand{\Z}{{\mathbb{Z}}}%
\newcommand{\C}{{\mathbb{C}}}%
\newcommand{\NN}{{\mathbb{N}}}%
\newcommand{\Al}{{\cal A}}%
\newcommand{\nonum}{\nonumber}%
\newcommand{\identitybf}{{\mathbf{1}} }
\newcommand{\idenA}{\identitybf_\Al}
\newcommand{\Mat}{{\mathrm{M}}}%
\newcommand{\Tr}{\mathbf{Tr}}%
\newcommand{\rmA}{{\mathrm{A}}}%
\newcommand{\rmB}{{\mathrm{B}}}%
\newcommand{\rmAB}{{\mathrm{AB}}}%
\newcommand{\I}{{\mathrm{I}}}%
\newcommand{\J}{{\mathrm{J}}}%
\newcommand{\IuJ}{{\mathrm{IJ}}}%
\newcommand{\K}{{\mathrm{K}}}%
\newcommand{\Kp}{\K^{\prime}}%
\newcommand{\Kx}{\K+x}%
\newcommand{\X}{{\mathrm{X}}}%
\newcommand{\Y}{{\mathrm{Y}}}%
\newcommand{\LL}{{\mathrm{L}}}%
\newcommand{\RR}{{\mathrm{R}}}%
\newcommand{\ZL}{\Z_\LL}%
\newcommand{\ZR}{\Z_\RR}%
\newcommand{\vpL}{\vp_{\ZL}}%
\newcommand{\vpR}{\vp_{\ZR}}%
\newcommand{\psiL}{\psi_{\LL}}%
\newcommand{\psiR}{\psi_{\RR}}%
\newcommand{\rhoL}{\varrho_{\LL}}%
\newcommand{\rhoR}{\varrho_{\RR}}%
\newcommand{\vrpL}{\varpi_{\LL}}%
\newcommand{\vrpR}{\varpi_{\RR}}%
\newcommand{\THEL}{\Theta_\LL}%
\newcommand{\THER}{\Theta_\RR}%
\newcommand{\omeL}{\ome_{\scriptscriptstyle \ZL}}%
\newcommand{\omeR}{\ome_{\scriptscriptstyle \ZR}}%
\newcommand{\Jone}{\J_1}%
\newcommand{\Jtwo}{\J_2}%
\newcommand{\rmAc}{\rmA^{c}}%
\newcommand{\Ic}{\I^{c}}%
\newcommand{\Alloc}{\Al_\circ}%
\newcommand{\AlI}{\Al({\I})}%
\newcommand{\AlJ}{\Al({\J})}%
\newcommand{\AlK}{\Al({\K})}%
\newcommand{\AlA}{\Al({\rmA})}%
\newcommand{\AlAc}{\Al({\rmAc})}%
\newcommand{\AlKx}{\Al({\K+x})}%
\newcommand{\AlIuJ}{\Al({\I}\cup {\J})}%
\newcommand{\AlIc}{\Al({\I}^{c})}%
\newcommand{\Ale}{\Al_{+}}%
\newcommand{\Alo}{\Al_{-}}%
\newcommand{\AlIe}{\Al(\I)_{+}}%
\newcommand{\AlIo}{\Al(\I)_{-}}%
\newcommand{\AlKe}{\Al(\K)_{+}}%
\newcommand{\AlJe}{\Al(\J)_{+}}%
\newcommand{\AlJo}{\Al(\J)_{-}}%
\newcommand{\AlZ}{\Al_{\Z}}%
\newcommand{\AlZloc}{\Al_{\Z \circ}}%
\newcommand{\AlL}{\Al_{\LL}}%
\newcommand{\AlR}{\Al_{\RR}}%
\newcommand{\AlLloc}{\Al_{\LL\,\circ}}%
\newcommand{\AlRloc}{\Al_{\RR\, \circ}}%
\newcommand{\AlLe}{\Al_{\LL +}}%
\newcommand{\AlLo}{\Al_{\LL -}}%
\newcommand{\AlRe}{\Al_{\RR +}}%
\newcommand{\AlRo}{\Al_{\RR -}}%
\newcommand{\AlZI}{\AlZ({\I})}%
\newcommand{\AlZL}{\AlZ({\ZL})}%
\newcommand{\AlZR}{\AlZ({\ZR})}%
\newcommand{\alL}{\alpha_{\LL}}%
\newcommand{\alR}{\alpha_{\RR}}%
\newcommand{\tilalL}{\widetilde{\alL}}%
\newcommand{\tilalR}{\widetilde{\alR}}%
\newcommand{\altL}{\alpha_{\LL,t}}%
\newcommand{\altR}{\alpha_{\RR,t}}%
\newcommand{\ome}{\omega}
\newcommand{\vp}{\varphi}
\newcommand{\vpIc}{\vp_{\Ic}}
\newcommand{\psiIc}{\psi_{\Ic}}
\newcommand{\FI}{\widetilde{F}_{\I}}
\newcommand{\FA}{\widetilde{F}_{\rmA}}
\newcommand{\FIpsi}{\FI(\psi)}
\newcommand{\FIvp}{\FI(\vp)}
\newcommand{\pot}{\Phi}%

\newcommand{\HI}{H_\I}%
\newcommand{\HIsur}{H_{\partial \I}}%
\newcommand{\HIopen}{\widetilde{H}_\I}%
\newcommand{\HLsur}{H_{\partial \ZL}}%
\newcommand{\HRsur}{H_{\partial \ZR}}%
\newcommand{\WLR}{W_{\LL, \RR}}%
\newcommand{\vpbW}{\vp^{\beta \WLR}}%
\newcommand{\vpbeta}{\vp_{\beta}}%
\newcommand{\vpg}{\vp_{\infty}}%
\newcommand{\HA}{H_\rmA}%
\newcommand{\HAsur}{H_{\partial \rmA}}%
\newcommand{\HAopen}{\widetilde{H}_\rmA}%
\newcommand{\ScI}{\widetilde{S}_{\I}}%
\newcommand{\ScIJ}{\widetilde{S}_{\I|\J}}%
\newcommand{\ScIJone}{\widetilde{S}_{\I|\Jone}}%
\newcommand{\ScIJtwo}{\widetilde{S}_{\I|\Jtwo}}%
\newcommand{\ScIIc}{\widetilde{S}_{\I|\Ic}}%
\newcommand{\ScAB}{\widetilde{S}_{\rmA|\rmB}}%
\newcommand{\ScA}{\widetilde{S}_{\rmA}}%
\newcommand{\potL}{\pot_\LL}%
\newcommand{\potR}{\pot_\RR}%
\newcommand{\potLR}{\pot_{{\LL, \RR}}}%

\newcommand{\potI}{\pot(\I)}%
\newcommand{\potK}{\pot(\K)}%
\newcommand{\potKp}{\pot(\Kp)}%
\newcommand{\potKx}{\pot(\Kx)}%
\newcommand{\potLK}{\potL(\K)}%
\newcommand{\potRK}{\potR(\K)}%
\newcommand{\Rnu}{\R^{\nu}}%
\newcommand{\Znu}{\Z^{\nu}}%
\newcommand{\tr}{\mathrm{tr}}%
\newcommand{\trI}{\tr_{\I}}%
\newcommand{\Fin}{{\mathfrak F}}%

\newcommand{\Finf}{\Fin_{{\rm{loc}} }}%
\newcommand{\FinfL}{\Fin_{\LL \, {\rm{loc}} }}%
\newcommand{\FinfR}{\Fin_{\RR \, {\rm{loc}} }}%
\newcommand{\BOX}{\Fin_{\rm{b}}}%
\newcommand{\LamlimJ}{\Lambda \nearrow \J}
\newcommand{\upright}{\upharpoonright}
\newcommand{\Hil}{{\cal H}}%
\newcommand{\pivp}{\pi_{\vp}}%
\newcommand{\Hilvp}{{\Hil}_{\vp}}%
\newcommand{\Ome}{\Omega}%
\newcommand{\Omevp}{\Ome_{\vp}}%
\newcommand{\vNM}{{\mathfrak{M}}}%
\newcommand{\vNMvp}{\vNM_{\vp}}%
\newcommand{\vph}{\vp^{h}}%
\newcommand{\omeh}{\omega^{h}}%
\newcommand{\vpmbh}{\vp^{-\beta h}}%
\newcommand{\sigtvp}{\sigma_t^{\vp}}%
\newcommand{\sigtvph}{\sigma_t^{[\vph]}}%
\newcommand{\Omevppih}{\Omevp^{\pivp(h)}}%
\newcommand{\Omevpk}{\Omevp^{k}}%
\newcommand{\rhopotI}{\rho^{\beta, \pot}_{\;\I}}
\newcommand{\vppotL}{\vp^{\beta, \potL}_{\;\LL}}
\newcommand{\vppotR}{\vp^{\beta, \potR}_{\;\RR}}
\newcommand{\vpbIsur}{\vp^{\beta\HIsur}}%
\newcommand{\gotimes}{\otimes_{\text{car}}}%
\newcommand{\cicr}{c_i^{\,\ast}}%
\newcommand{\ci}{c_i}%
\newcommand{\cjcr}{c_j^{\,\ast}}%
\newcommand{\cj}{c_j}%
\newcommand{\cicrsig}{c_{i\sigma}^{\,\ast}}%
\newcommand{\cisig}{c_{i\sigma}}%
\newcommand{\cjcrsigp}{c_{j {\sigma}^\prime}^{\,\ast}}%
\newcommand{\cjsigp}{c_{j {\sigma}^\prime}}%
\newcommand{\OA}{\mathcal{O}_\rmA}%
\newcommand{\OB}{\mathcal{O}_\rmB}%
\newcommand{\altpot}{\alpha_{\pot,t}}%
\newcommand{\altpoth}{\altpot^{h}}%
\newcommand{\altpotmW}{\altpot^{-\WLR}}%
\newcommand{\almbtpot}{\alpha_{\pot,-\beta t}}%
\newcommand{\altpotL}{\alpha_{\potL,t}}%
\newcommand{\altpotR}{\alpha_{\potR,t}}%
\newcommand{\SEP}{\mathfrak{S}}
\begin{document}
\title[Mutual entropy and thermal area law]{Mutual entropy and 
thermal area law in  $C^{\ast}$-algebraic quantum lattice systems}
%

\author*{\fnm{Hajime} \sur{Moriya}}\email{hmoriya4@se.kanazawa-u.ac.jp}
\affil*{\orgdiv{Institute of Science and Engineering}, \orgname{Kanazawa University}, \orgaddress{\street{Kakuma-Machi}, \city{Kanazawa}, \postcode{920-1192}, \state{Ishikawa}, \country{Japan}}}
\abstract{We present a general definition of mutual entropy for 
infinitely extended quantum spin  and fermion lattice systems, 
and show  its fundamental  properties.
Using the mutual entropy, we establish a thermal area law 
in these infinitely extended quantum systems.
 The proof is based on the local thermodynamical stability (LTS) 
  formulated as a variational principle 
in terms of the conditional free energy  on local subsystems.
 Our  thermal area law in quasi-local $C^{\ast}$-systems    
  applies to general interactions with well-defined surface energies.
Furthermore, we examine the mutual entropy between the left- and right-sided infinite regions of one-dimensional lattice systems. For general translation-invariant finite-range interactions on such systems, the thermal equilibrium state at any temperature exhibits a finite value of the mutual entropy 
between these infinite disjoint regions. 
This result implies that the infinitely large quantum entanglement 
characteristic of critical ground states in one-dimensional systems 
is  drastically destroyed by even a small  positive temperature, 
indicating thermal suppression of quantum entanglement.}


\maketitle
\tableofcontents
\section{Introduction}
\label{sec:INTRO}
The main purpose of this paper is to establish a thermal area law for infinitely extended  quantum  lattice systems. 
In Subsection~\ref{subsec:introFINITE}, we recall the  thermal area law 
in the finite-dimensional setting, as  presented in \cite{WOLF},
 and introduce  basic notation which will be used throughout this paper.  
 In Subsection~\ref{subsec:Purpose}, we state  our objective 
for investigating the thermal area law in a $\cstar$-algebraic framework.

\subsection{Thermal area law in finite-dimensional systems}
\label{subsec:introFINITE}
  Consider a compound quantum system  on some underlying space $\Gamma$.
For any subset $\rmA \Subset \Gamma$, 
the subsystem  associated with ${\rmA}$
 is given  by  a finite-dimensional matrix algebra. 
Let $\rho$ be  a  state  on  $\Gamma$.  
 The  reduced  state  $\rho_{\rmA}$ is the restriction of $\rho$
 to the subsystem on $\rmA$.
We  will occasionally  use the slightly heavy 
notation $\rho\!\!\upright_{\rmA}$ instead of 
$\rho_{\rmA}$ in order to  emphasize the restriction
 of the global state  $\rho$  to the subsystem on $\rmA$.

The von Neumann entropy of a state $\rho$ on $\rmA$ is defined  as  
\begin{equation}
\label{eq:vonENT}
S_{\rmA}(\rho)\equiv -\Tr(D_{{\rho}_{\rmA}}\log D_{{\rho}_{\rmA}}),
\end{equation}
 where   $D_{{\rho}_{\rmA}}$ denotes the density matrix  corresponding 
 to the state $\rho_{\rmA}$  with respect to the 
matrix trace $\Tr$. In contrast to the tracial state $\tr$,   
 the matrix trace $\Tr$ takes $1$ on   
  each one-dimensional projection, and therefore  
  $\frac{1}{n_{\rmA}}\Tr_{\rmA}= \tr_{\rmA}$ holds, where 
$n_{\rmA}$ denotes  the matrix dimension of the subsystem.

Next, we  recall the 
 quantum relative entropy  \cite{UME62}.
 For  two states $\rho$ and $\sigma$, 
the  quantum relative entropy of them on 
  ${\rmA} \Subset \Gamma$ is given by 
\begin{equation}
\label{eq:UMEGAKI}
S(\rho_{\rmA} \mid  \sigma_{\rmA})\equiv
\Tr\left(D_{{\rho}_{\rmA}}(\log D_{{\rho}_{\rmA}}
-\log D_{{\sigma}_{\rmA}}) \right).
\end{equation}
The  connection between 
 the von Neumann entropy and the quantum  relative entropy
is as follows:
\begin{align}
\label{eq:rel-and-ent}
S(\rho_{\rmA} \mid \tr_{\rmA})= -S_{\rmA}(\rho)+\log n_{\rmA}. 
\end{align}

Consider  any two disjoint subsets  $\rmA, \rmB \Subset \Gamma$.
For a state $\rho$, the conditional entropy of $\rmA$ given 
 the condition $\rmB$ is defined as 
\begin{equation}
\label{eq:CE-finite}
\ScAB(\rho):= S_{\rmAB}(\rho)-S_{\rmB}(\rho). 
\end{equation}
It  is often denoted  $H(\rmA | \rmB)$ in information theory. 

The mutual entropy of a  state $\rho$
between  disjoint subsets  ${\rmA}$ and ${\rmB}$ is given 
  by
\begin{equation}
\label{eq:MUT}
I_\rho(\rmA:\rmB):=
S_{\rmA}(\rho)+S_{\rmB}(\rho)-S_{\rmAB}(\rho).
\end{equation}
 The mutual entropy is also expressed in terms of
the conditional entropy  as
\begin{equation}
\label{eq:MUT-CondExp-finite}
I_\rho(\rmA:\rmB)=
S_{\rmA}(\rho)-\ScAB(\rho).
\end{equation}
The mutual entropy has another 
notable  expression in terms of 
 the quantum relative entropy as 
\begin{equation}
\label{eq:MIrelative}
I_\rho(\rmA:\rmB)=
S(\rho_{\rmAB} \mid \rho_{\rmA}\otimes\rho_{\rmB}).
\end{equation}

We now turn to  statistical-mechanical considerations.
Take a pair of disjoint  regions  ${\rmA}$ and ${\rmB}$ as above. 
The region $\rmA$  represents a subsystem  of interest, 
whereas its exterior region $\rmB$ lies in $\rmAc$, 
the complement  of $\rmA$. 
Let $H_{\rmAB}$ denote a Hamiltonian of the quantum system on $\rmAB(\equiv \rmA \cup\rmB)$, which can be decomposed as
\begin{equation}
\label{eq:Hdecompop}
H_{\rmAB}=H_{\rmA}+H_{\partial \rmA} + H_{\rmB},
\end{equation}
where $H_{\rmA}$ and $H_{\rmB}$ are  local Hamiltonians on  
the  specified regions ${\rmA}$ and ${\rmB}$, respectively, and
$H_{\partial \rmA}$ denotes  the interaction  
  between $\rmA$ and $\rmB$.  
The term $H_{\partial \rmA}$ is commonly referred to 
as the surface energy.
The region $\partial \rmA$ denotes the support  
of the local operator $H_{\partial \rmA}$,  
  corresponding  to the boundary \emph{area} 
 between  $\rmA$ and $\rmB$, which 
intersects  both regions. 
This boundary area  will play a central role 
in  the thermal area law, which  will be introduced below.
 
For  each  inverse temperature  $\beta>0$, the  
 Gibbs  state $\rho_{\text{Gib}\, \rmAB}^\beta$ associated with  
 the Hamiltonian $H_{\rmAB}$ is 
 defined  by 
\begin{equation}
\label{eq:GibbsAB}
\rho_{\text{Gib}\, \rmAB}^\beta(X):=\frac{{\rm Tr} \left( e^{-\beta H_{\rmAB}} X \right) }
{{\rm Tr}\left( e^{-\beta H_{\rmAB}}\right)}, 
\end{equation}
where $X$ denotes an arbitrary operator on the region $\rmAB$.  

The free energy functional of  any  state $\rho$ 
on  $\rmAB$ is  defined as  
\begin{equation}
\label{eq:FE}
F(\rho)\equiv {\rm Tr}(H_{\rmAB}\rho)-\frac{1}{\beta}
S_{\rmAB}(\rho). 
\end{equation}
 It is well known  that the Gibbs state \eqref{eq:GibbsAB} defined above 
 minimizes the free energy  among all states of  the system on $\rmAB$, 
see  Section 5 of \cite{WEHRL}, for example. 
In particular, 
\begin{equation}
\label{eq:FEminiplug}
F(\rho_{\text{Gib}\, \rmAB}^\beta)
\le F(\rho_{\text{Gib}\, \rmAB}^\beta\!\!\upright_{{\rmA}}
 \otimes \rho_{\text{Gib}\, \rmAB}^\beta\!\!\upright_{\rmB}), 
\end{equation}
where the  state on  the right-hand side   
 is  the product of  the reduced states 
 of 
 $\rho_{\text{Gib}\, \rmAB}^\beta$ on $\rmA$ and $\rmB$.
From this  inequality  it follows that 
\begin{equation}
 \label{eq:WOLF-finite-TAL} 
I_{\rho_{\text{Gib}\, \rmAB}^\beta} (\rmA:\rmB)\le \beta
(\rho_{\text{Gib}\, \rmAB}^\beta\!\!\upright_{{\rmA}} \otimes
\rho_{\text{Gib}\, \rmAB}^\beta\!\!\upright_{{\rmB}}-
\rho_{\text{Gib}\, \rmAB}^\beta)\left(H_{\partial \rmA}\right) \le 
2 \beta \lVert H_{\partial \rmA} \rVert.
\end{equation}
 Note that 
the reduced states  
$\rho_{\text{Gib}\, \rmAB}^\beta\!\!\upright_{{\rmA}}$ and 
 $\rho_{\text{Gib}\, \rmAB}^\beta\!\!\upright_{{\rmB}}$ used above 
 are 
different from  the  Gibbs states determined by 
the local Hamiltonians $H_{\rmA}$ and $H_{\rmB}$, respectively.

If  the surface energy is estimated by the area of the surface as   
\begin{equation}
 \label{eq:HpartialA} 
\lVert H_{\partial \rmA} \rVert \le c|\partial \rmA| 
\end{equation}
 for some constant $c>0$, then \eqref{eq:WOLF-finite-TAL} 
 yields 
\begin{equation}
\label{eq:Wolf-geometric} 
I_{\rho_{\text{Gib}\, \rmAB}^\beta}(\rmA: \rmB)\le  2\beta c 
| {\partial \rmA}|.
\end{equation}
This is the familiar expression of the thermal area law, which depends 
on the surface area $|\partial \rmA|$ rather than on the volume $|\rmA|$. 
The thermal area law obtained in this way holds for quantum systems on an infinite-dimensional Hilbert space, provided that the local Gibbs states are represented by density matrices (positive trace-class operators); see \cite{LEM}.

\subsection{Thermal area law in the $C^{\ast}$-algebraic framework}
\label{subsec:Purpose}
We aim to formulate a thermal area law for quantum spin lattice systems and 
fermion lattice systems adopting the $\cstar$-algebraic framework.
We explain our motivation.

The thermal area law as presented  in \cite{WOLF} and  
summarized in Subsection~\ref{subsec:introFINITE}
 uses the Hilbert-space formalism, in particular the so-called  
\emph{box procedure}. 
This conventional approach of statistical mechanics  
 is based on local Gibbs states associated with definite local Hamiltonians, 
each defined under a \emph{specific} boundary condition 
or a hypothetical wall enclosing a finite region (box).  

While the box procedure serves as a handy and practical formulation, 
it has not been proved that all equilibrium states 
 can be thoroughly exhausted within this procedure, 
except for a few known cases.  From a 
mathematically  rigorous standpoint, it is certainly a limitation.
Moreover, in the formulation of the thermal area law that we now argue, 
there is another subtle aspect.  
In the box procedure, 
each local Gibbs state depends on two finite regions as its parameters: 
a finite subregion $\rmA$ representing the local system of interest, and another finite subregion $\rmB$ representing a thermal bath 
coupled to $\rmA$. This two-region dependence is reflected 
in the conventional form of the thermal area law  presented 
 in \eqref{eq:Wolf-geometric}.

However, treating these two finite regions in a consistent  manner 
 is not straightforward. In particular,  how to take 
the appropriate infinite-volume limit of such double-indexed local Gibbs 
states remains ad hoc unless supplemented with  specific physical input.

In light of the pioneering works on the area law
  \cite{SORKIN, SRED}, it is essential to consider reduced (partial) states 
 of a global state  defined on an infinitely extended space. 
Here, the notion of modular Hamiltonians naturally emerges;   
 see, for instance, \cite{CAS-HUE-LEC} and \cite{CARDY}. 

The $\cstar$-algebraic framework provides a natural setting for describing
 such infinitely extended quantum systems, where
all local subsystems are embedded.
 Hence, we employ the $\cstar$-algebraic framework
 to formulate  a thermal area law for infinitely extended quantum 
(lattice) systems.  We consider that this is more than a simple infinite-dimensional reformulation of the known result. 
For general discussion of the $\cstar$-algebraic approach to quantum statistical mechanics in comparison with the box procedure,   
 we refer to the introduction of \cite{BRA2}, and Section 2 of \cite{ARAKI78}.

\section{$C^{\ast}$-algebraic quantum lattice systems}
\label{subsec:QSYSTEMS}
In this section, we briefly introduce the basic formalism of 
$\cstar$-algebraic quantum lattice systems. 
We refer to  \cite{BRA2} as a standard reference.

Let $\Gamma$ denote an  infinite lattice. 
For example, $\Gamma$ can be a $\nu$-dimensional 
cubic lattice $\Znu$ with $\nu \in \NN$. Let $\Al$ denote  
 a quantum spin lattice system or a fermion lattice system on $\Gamma$.
Precisely, $\Al$ is a quasi-local $\cstar$-system on  $\Gamma$ given as follows. 
Let $\idenA$ denote the unit of $\Al$.
 Let $\Fin$ denote the  set of all subsets of $\Gamma$. 
If $\I \in\Fin$ has  finite cardinality (finite volume) $|\I|<\infty$,
 then we denote  $\I \Subset \Gamma$.
Let  $\Finf$ denote the set of all finite subsets of $\Gamma$.
For each $\I\in \Finf$, the subsystem $\AlI$
 is a finite-dimensional matrix algebra. 
The local algebra  $\Alloc:=\bigcup_{\I \in\Finf}\AlI$ 
is  norm  dense in the total  $\cstar$-system $\Al$.

Thus far, we have introduced  the common structure 
 of quantum lattice systems.
In the following,
we distinguish between quantum spin lattice systems
and fermion lattice systems,
which are characterized by tensor-product structures
and the canonical anticommutation relations (CAR), respectively.

\subsection{Quantum spin lattice systems}
\label{subsec:SPIN}
Quantum spin lattice systems have the following local structure.
For each $\I \in \Finf$, $\AlI$ is isomorphic to a 
 full matrix algebra  $\Mat_{k}(\C)$ for some  $k\in \NN$.
For  disjoint subsets $\I, \J \in \Finf$, the joint system $\AlIuJ$ is given by the tensor product of 
$\AlI$ and $\AlJ$:
\begin{equation}
\label{eq:TENSOR-IJ} 
\AlIuJ=\AlI\otimes \AlJ.
\end{equation}

\subsection{Fermion lattice systems}
\label{subsec:FERMIONsys}
Let $\ci$ and $\cicr$ denote the annihilation 
and  creation operators of a fermion at site $i\in \Gamma$, respectively. 
They satisfy the  canonical anticommutation relations (CAR): 
\begin{align}
\label{eq:CAR}
\{ \cicr, \cj \}&=\delta_{i,j}\, \idenA, \nonumber \\
\{ \cicr, \cjcr \}&=\{ \ci, \cj \}=0.
\end{align}
For each $\I\in \Finf$, $\AlI$ is given by the finite-dimensional 
algebra generated by $\{\cicr, \, \ci\, ;\;i\in \I\}$, 
which is  isomorphic to $\Mat_{2^{|\I|}}(\C)$.

Let $\Theta$ denote the involutive automorphism on the fermion system  
 $\Al$  determined  by 
 \begin{equation}
\label{eq:CAR-THETA}
\Theta(\ci)=-\ci, \quad \Theta(\cicr)=-\cicr,\quad i\in \Gamma.
\end{equation}
The grading structure on $\Al$ is given by $\Theta$ as:
\begin{equation}
\label{eq:CAREO}
 \Ale := \{A \in \Al \; \bigl| \;   \Theta(A)=A  \},\quad  
 \Alo := \{A \in \Al  \; \bigl| \;  \Theta(A)=-A  \},
\end{equation}
\begin{equation}
\label{eq:CARbunkai}
\Al=\Ale\oplus\Alo.
\end{equation}
For each $\I\in \Fin$, let 
\begin{equation}
\label{eq:CARIeo}
 \AlIe := \AlI\cap \Ale,\quad  \AlIo := \AlI\cap \Alo,
\end{equation}
\begin{equation}
\label{eq:CARIbunkai-i}
\AlI=\AlIe\oplus\AlIo.
 \end{equation}
If a state $\rho$ on $\Al$ is invariant under the fermion grading $\Theta$, 
then it vanishes on $\Alo$ and is called an even state.

By \eqref{eq:CAR}, for any disjoint pair of regions 
$\I, \J \in \Fin$, the $\Theta$-graded locality holds:
\begin{align}
\label{eq:glocality}
[A_{+},\; B_{+}]&=0\ {\text{for}}\ A_{+} \in \AlIe, \ 
B_{+} \in \AlJe, \nonum\\  
[A_{+},\; B_{-}]&=0\ {\text{for}}\ A_{+} \in \AlIe, \ 
 B_{-} \in \AlJo, \nonum\\  
[A_{-},\; B_{+}]&=0\ {\text{for}}\ A_{-} \in \AlIo, \ 
 B_{+} \in \AlJe, \nonum\\  
\{A_{-},\; B_{-}\}&=0\ {\text{for}}\  A_{-} \in \AlIo, \ 
 B_{-} \in \AlJo.
\end{align}
Let $\theta$ be a $\pm 1$-valued symmetric function on even-odd elements 
 in two disjoint regions given as 
\begin{align}
\label{eq:theta}
1&=
\theta (A_{+}, B_{+})=\theta ( B_{+}, A_{+})=
\theta (A_{+}, B_{-})=\theta ( B_{-}, A_{+})=
\theta (A_{-}, B_{+})=\theta ( B_{+}, A_{-}), \nonum\\
-1&=\theta (A_{-}, B_{-})=\theta ( B_{-}, A_{-}).
\end{align}
By using the function $\theta$, the graded commutation relations 
  \eqref{eq:glocality} can be rewritten in the following  compact form: 
\begin{align}
\label{eq:compact}
A_{\sharp}B_{\flat}= \theta (A_{\sharp}, B_{\flat}) B_{\flat} A_{\sharp},
\quad  A_{\sharp} \in \AlIe \ \text{or} \ \in \AlIo, \ 
 B_{\flat} \in \AlJe \ \text{or} \ \in \AlJo.
\end{align}

We may consider fermions  with  finitely many 
  spin degrees of freedom, labeled by $\sigma$. 
These fermions obey the canonical anticommutation relations:  
\begin{align}
\label{eq:CARspin}
\{ \cicrsig, \cjsigp \}&=\delta_{i,j} \delta_{\sigma, \sigma^{\prime}}
\, \idenA, \nonumber \\\{ \cicrsig, \cjcrsigp \}&=\{ \cisig, \cjsigp \}=0.
\end{align}
Since this generalization does not affect the argument  to be presented, 
 we deal with   spinless fermion systems as in  
  \eqref{eq:CAR} to avoid unnecessary notational clutter.

\section{Quantum mutual entropy for infinitely extended quantum lattice systems}\label{sec:MUT-Def}
To formulate a thermal area law in the $\cstar$-algebraic framework, we need 
 the notion of quantum mutual entropy.
Recently, in algebraic quantum field theory (AQFT),  studies  related to 
 quantum mutual entropy
 have  been developed; see e.g. \cite{HO-SA}. 
It seems, however, that   
 a general and  systematic  treatment of the quantum mutual entropy 
in $\cstar$-algebraic quantum statistical mechanics is  scarce. 
 See Remark~\ref{rem:MUTUAL} below. 

In this section, we define the mutual entropy  
 in quasi-local $\cstar$-systems representing  
 quantum spin lattice systems and fermion lattice systems, 
 and provide its basic properties.

\begin{remark}
\label{rem:MUTUAL}
 In the seminal work \cite{LIND73},
the quantum mutual entropy as in \eqref{eq:MUT} was introduced  
  for finite-dimensional quantum systems.
The standard reference on $\cstar$-algebraic quantum statistical mechanics \cite{BRA2} does not directly address the mutual entropy within this framework. 
The extensive monograph on quantum entropy 
 \cite{OHYA-PETZ}, contrary to  expectation, 
does not present  a $\cstar$-algebraic (operator-algebraic)
 extension of the mutual entropy. 
Instead, \cite{OHYA-PETZ} introduces other elaborate quantities 
under the term ``quantum mutual entropy,'' which are primarily intended for the study of quantum channels. 
\end{remark}

In the following subsections, we  introduce basic  entropy functionals--
von Neumann entropy,  conditional entropy, and mutual entropy--  
within the quasi-local $\cstar$-algebras,   
and provide  their  basic properties  required for our purpose.
There is no essential  distinction  between the quantum spin lattice system and the fermion lattice system. However, certain subtleties will arise when considering  general (non-even)  states on the fermion system.

\subsection{von Neumann entropy}
\label{subsec:vonENT}
We briefly recall the von Neumann entropy and its properties, 
which serve as the basis for defining conditional entropy and mutual entropy.

Consider an arbitrary  state  $\psi$ on the quasi-local $\cstar$-system $\Al$. 
The von Neumann entropy $S_{\I}(\psi)$ of $\psi$ on $\I\in\Finf$
is defined as in \eqref{eq:vonENT}. 
It satisfies  the strong subadditivity (SSA) property:
 For  $\X, \Y\in \Finf$,
\begin{align}
\label{eq:SSA}
S_{\X \cap \Y}(\psi)+S_{\X \cup\Y}(\psi)\le 
S_{\X}(\psi)+S_{\Y}(\psi).
\end{align}
SSA is a fundamental  property of the von Neumann entropy, 
proved by Lieb and Ruskai \cite{LIEBRUSKAI73}.
SSA also holds for  fermion lattice systems without any restriction on states  as shown  in \cite{SSA}.

\subsection{Conditional entropy}
\label{subsec:CONDENT}
We now introduce the conditional entropy, following 
 Section 6 of \cite{RIMSI76}; see also  Definition 6.2.27 of \cite{BRA2}.
Let $\psi$ be an arbitrary state of $\Al$. Take any $\I\in \Finf$. 
Let  $\J \subset \Ic$, which  can be  either finite or infinite.
The conditional entropy of $\psi$ on $\I$ given  $\J$ is defined by 
\begin{align}
\label{eq:ScIJpsi}
\ScIJ(\psi)&:=
\inf_{\Lambda \Subset \J} \bigl\{ S_{\I\cup \Lambda}(\psi)
-S_{\Lambda }(\psi) \bigr\}\nonum \\
&=\lim_{\LamlimJ} \bigl\{ S_{\I\cup \Lambda}(\psi)
-S_{\Lambda }(\psi) \bigr\}.
\end{align}
The existence of the limit as an infimum is guaranteed by the strong subadditivity of the von Neumann entropy \eqref{eq:SSA} as stated in Proposition 6.2.26 of \cite{BRA2}. 
If $\J$ is  finite, then it coincides with  the formula 
$\ScIJ(\psi)= S_{\I \cup \J}(\psi)-S_{\J}(\psi)$ given in \eqref{eq:CE-finite}. If $\J=\emptyset$, then it is reduced to the von Neumann entropy $S_{\I}(\psi)$ When $\J =\Ic$, the corresponding conditional entropy
$\ScIIc(\psi)$ will be denoted by $\ScI(\psi)$ as in \cite{BRA2}.
For each fixed $\I\in \Finf$, $\ScIJ(\psi)$ is a non-increasing function of $\J\subset \Ic$ with respect to inclusion  as noted  in Proposition 6.2.25 of \cite{BRA2}. Namely, for $\Jone \subset \Jtwo \subset \Ic$
\begin{align}
\label{eq:MONOTONEcondENT}
\ScI(\psi)\le \ScIJtwo(\psi)\le
\ScIJone(\psi)\le S_\I(\psi).  
\end{align}
Furthermore, for any state of the quantum spin lattice system and any \emph{even} state of the fermion lattice system, the inequality 
\begin{equation}
\label{eq:ineq-BOTH}
 \left| \ScIJ(\psi) \right| \le S_{\I}(\psi)
\end{equation}
holds for all $\J\subset \Ic$. 
As noted in Proposition 6.2.25 of \cite{BRA2}, it follows from 
 the triangle inequality  of the  von Neumann entropy \cite{Araki-Lieb} \cite{VALIDITY}.
Note, however, that some non-even states of the fermion system fail to  satisfy  \eqref{eq:ineq-BOTH}; 
see \cite{SOME, VALIDITY} for explicit  counterexamples.

\subsection{Mutual entropy}
\label{subsec:MUTUALquasilocal}
We need the mutual entropy on quantum spin and fermion lattice systems in the case where one of the disjoint regions is finite. 
This corresponds to the standard  setup of  the thermal area law,
which will be discussed in Section~\ref{sec:TAL}.
So throughout   this subsection, we assume that the region $\I$ is  finite, 
while  the other region $\J$  
 in the complement of $\I$  can be  either  finite or infinite.
Later in Section~\ref{sec:ONE}, we discuss the case where  
  both disjoint regions $\I$ and $\J$  are infinite.

We shall formulate the  mutual entropy in terms of the conditional entropy, rather than  the quantum relative entropy. This somewhat indirect definition is designed to accommodate general states
 which need  not be  modular (faithful) states; see Subsection \ref{subsec:GR}.
It also enables us to treat the fermion system in full generality.

Let $\psi$ be an arbitrary state on the quasi-local $\cstar$-system $\Al$.  
Consider two disjoint  regions $\I\in \Finf$ and   $\J \in \Fin$. 
The mutual entropy of $\psi$ between $\I$ and $\J$ is defined by
\begin{equation}
\label{eq:MUTUAL-IJ}
I_\psi(\I:\J):=
S_{\I}(\psi)-\ScIJ(\psi), 
\end{equation}
in particular, 
\begin{equation}
\label{eq:MUTUAL-IandIc}
I_\psi(\I:\Ic):=S_{\I}(\psi)-\ScI(\psi).
\end{equation}
If $\J$ is finite, then this reduces to  the  finite-dimensional formula
 $I_\psi(\I:\J)=S_{\I}(\psi)+S_{\J}(\psi)-S_{\IuJ}(\psi)$ given in 
  \eqref{eq:MUT}.

By the inequality \eqref{eq:MONOTONEcondENT}, the mutual entropy
 is non-negative: 
\begin{equation}
\label{eq:MUTUAL-IJ-POSITIVE}
0 \le I_\psi(\I:\J).
\end{equation}
 For each fixed $\I\in \Finf$, by \eqref{eq:MONOTONEcondENT}, 
the mutual entropy is 
monotone with respect to inclusion  of  the outside region. Namely, 
 for  $\Jone \subset \Jtwo \subset \Ic$,  
\begin{equation}
\label{eq:MUTUAL-MONOTONE}
I_\psi(\I:\Jone)\le  I_\psi(\I:\Jtwo).
\end{equation}

By the estimate \eqref{eq:ineq-BOTH}, 
for an arbitrary state $\psi$ of the quantum spin lattice system 
 and  an arbitrary even state $\psi$  of  the fermion lattice system, 
the mutual entropy on any fixed finite $\I$ is 
 bounded by twice the von Neumann entropy: 
For any   $\J\subset\Ic$, 
\begin{equation}
\label{eq:MUT-2von}
I_\psi(\I:\J)\le  2S_{\I}(\psi).
\end{equation}
 This inequality is well known in the 
finite-dimensional case.
Again, note that 
some  non-even states of the fermion system 
 invalidate  \eqref{eq:MUT-2von} as shown in \cite{SOME, VALIDITY}.

\subsection{Mutual entropy in terms of quantum relative entropy}
\label{subsec:MUT-ARAKI}
We now reformulate the mutual entropy defined 
 in \eqref{eq:MUTUAL-IJ} in terms of 
the quantum relative entropy as in  \eqref{eq:MIrelative}. 

The quantum relative entropy of two states  $\omega$ and $\varrho$ on the $\cstar$-system $\Al$ is \emph{formally} given by  
\begin{equation}
\label{eq:Q-R-mugen}
S(\omega \mid \varrho)=\omega \left(\log \omega-\log \varrho\right).
\end{equation}
To make this expression rigorous, we assume that 
 both $\omega$ and $\varrho$ are modular (faithful) states. 
The definition of modular states will be  given 
 in Definition~\ref{df:MODULAR} in Section~\ref{sec:EQUILIBRIUM-STATES}. 
We then apply Araki's definition of quantum relative entropy \cite{RIMSI76, RIMSII77} to these two states,  
\begin{equation}
\label{eq:ARAKI-relative}
 S(\omega \mid \varrho)\equiv S_{\rm{ARAKI}}(\varrho/\omega )
:=-(\Psi_{\omega}, \log \Delta_{\varrho, \omega} \Psi_{\omega}), 
\end{equation}
 where $\Delta_{\varrho, \omega}$ denotes the relative modular operator.
Precisely, one takes GNS representations of the states and applies the formula 
\eqref{eq:ARAKI-relative} in the setting of von Neumann algebras, as in  Lemma 3.1 of \cite{HOT83}. We also refer to Appendix of \cite{BOST} for this technical point.

\begin{remark}
\label{rem:ARAKI-UME}
For Araki's quantum relative entropy,  
 we adopt Umegaki's notation $S(\omega \mid \varrho)$ \cite{UME62} as above, 
since this  notation has been  widely  used in  the literature; we refer to
 some reviews \cite{CAS-HUE-LEC, HO-SA, WITT}. However, in previous works \cite{ASEW, AM-LTS} on the LTS condition, which is another key concept in the present paper, Araki's  notation was employed. 
\end{remark}

In this subsection, let $\psi$ denote  an arbitrary  modular (faithful) state of $\Al$. For the quantum spin lattice system $\Al$, 
 the conditional entropy of $\psi$ on $\I\in \Finf$ 
can be  expressed in terms of Araki's quantum relative entropy as
\begin{align}
\label{eq:ScI-relative-form}
\ScI(\psi)=-S(\psi \mid \trI\otimes \psi_{\Ic})+\log n_{\I},
\end{align}
where $\trI\otimes \psi_{\Ic}$ denotes the product  
 of the tracial state $\trI$ on $\AlI$ and 
 the reduced  state of $\psi$  to $\AlIc$, and $n_{\I}$
 is the matrix dimension of the subsystem  $\AlI$; see \cite{ASEW} for details.
Analogously,  for the fermion  lattice system $\Al$,
 Proposition 7 of \cite{AM-LTS} shows that  
the conditional entropy of $\psi$ on $\I\in \Finf$ can be expressed  as
\begin{align}
\label{eq:FER-ScI-relative-form}
\ScI(\psi)=-S(\psi \mid \trI \gotimes \psi_{\Ic})+\log n_{\I},
\end{align}
where $\trI\gotimes \psi_{\Ic}$  denotes the 
 product-state extension  of the tracial state $\trI$ on $\AlI$ and 
 the reduced state of $\psi$ to  $\AlIc$. 
Note that $\psi$ is not necessarily even.

Next, we turn to  the mutual entropy. For the quantum spin system, 
by \eqref{eq:ScIJpsi}, \eqref{eq:MUTUAL-IJ}, and \eqref{eq:MIrelative}, 
the mutual entropy of a modular state $\psi$ can be  rewritten 
in terms of  Araki's quantum relative entropy as 
\begin{align}
\label{eq:MUT-IJ-RELara}
I_\psi(\I:\J)
&=
\lim_{\LamlimJ} \bigl\{ S_{\I}(\psi)+S_{\Lambda }(\psi)
-S_{\I \cup \Lambda}(\psi)\} \nonum\\
&=\lim_{\LamlimJ}  
S(\psi_{\I \cup \Lambda} \mid \psi_{\I}\otimes \psi_{\Lambda})
 \nonum\\
&=S(\psi_{\IuJ} \mid \psi_{\I}\otimes \psi_{\J})<\infty, 
\end{align}
where the convergence follows from the monotonicity of Araki's quantum 
relative entropy \cite{RIMSI76} with respect to inclusion of subsystems 
and the uniform boundedness 
 \eqref{eq:MUT-2von}.
For the fermion system, 
  assuming additionally evenness of  $\psi$, 
we obtain 
\begin{align}
\label{eq:FER-MUT-IJ-RELara}
I_\psi(\I:\J)
=S(\psi_{\I\cup \J} \mid \psi_{\I}\gotimes \psi_{\J})<\infty
\end{align}
by the same reasoning as in  \eqref{eq:MUT-IJ-RELara}.
Note that if $\psi$ is non-even, the product extension  
$\psi_{\I}\gotimes \psi_{\J}$
may not  exist  as noted in \cite{SOME}, and 
 the above expression 
\eqref{eq:FER-MUT-IJ-RELara} does not hold.

Comparing 
\eqref{eq:ScI-relative-form} and  
\eqref{eq:FER-ScI-relative-form}
 with  \eqref{eq:MUT-IJ-RELara} and \eqref{eq:FER-MUT-IJ-RELara}, 
we see that the conditional entropy is a special 
  mutual entropy (up to some additive constants).

\section{Thermal equilibrium}
\label{sec:EQUILIBRIUM-STATES}
There are various 
  characterizations of thermal equilibrium 
 in the $\cstar$-algebraic formulation \cite{BRA2}.
 In this paper,  we  use the  local thermodynamical stability, 
 the Gibbs condition, and the KMS condition. 
  While the well-known KMS condition plays crucial   
 roles in several points  in this paper, we adopt
the local thermodynamical stability (LTS) as our primary notion 
 of thermal equilibrium. 
Throughout this paper, the symbol $\vp$  denotes an arbitrary thermal 
 equilibrium state at positive  temperature.
It is not necessarily a factor state (i.e., pure phase).

\subsection{Local thermodynamical stability}
\label{subsec:LTS}
We  recall the  local thermodynamical stability (LTS) 
condition in a unified manner
  for both quantum spin lattice 
systems \cite{ASEW} and  fermion lattice systems \cite{AM-LTS}.

A potential is a map $\pot:\Finf\to\Alloc$ such that
\begin{equation}
\label{eq:POT-hermite}
\potK^*=\potK\in\AlK,\qquad \K\in\Finf.
\end{equation}
For fermion  lattice systems, 
we assume, in accordance with the locality principle, that 
 every  $\potK$ ($\K\in \Finf$) is  even:  
\begin{equation}
\label{eq:POT-even}
\potK^{\ast}=\potK\in \AlKe.
\end{equation}
Thus, local commutativity  holds  
for both quantum spin and fermion lattice systems: 
\begin{equation}
  [\potK ,\; \potKp]=0  \quad \text{if}\ 
 \K \cap \Kp =\emptyset\ (\K, \Kp\in \Finf).
\end{equation}
Translation invariance is not required for $\pot$.

For each $\I\in \Finf$, the inner local   Hamiltonian  is given as  
\begin{equation}
\label{eq:HI}
\HI:= \sum_{\K:\; \K\subset \I} \potK\in \AlI.
\end{equation}
For each $\I\in \Finf$, 
the surface energy is assumed to exist as an  element of $\Al$ 
\begin{equation}
\label{eq:HIsur}
\HIsur:=
 \sum_{\K:\; \K\cap  \I\ne \emptyset, 
 \K\cap  \Ic\ne \emptyset}
 \potK\in \Al.
\end{equation}
 $\HIsur$ does not necessarily belong to $\Alloc$, as
 its support ${\partial \I}$ may be infinite.
Set
 \begin{equation}
\label{eq:HIopen}
\HIopen:=\HI+\HIsur \in \Al.
\end{equation}

For each $\I\in\Finf$, the local Gibbs state on $\AlI$ 
 at inverse temperature $\beta$ with respect to the potential $\pot$
 is defined by  
\begin{equation}
\label{eq:locGib-I}
\rhopotI(A)
:=\dfrac{1}{{\rm Tr}\left(\exp(-\beta \HI)\right)}{\rm Tr}
\left(\exp(-\beta \HI)A\right), 
\quad A \in \AlI. 
\end{equation}
These local Gibbs states, determined by the \emph{inner} (free-boundary) 
 local Hamiltonians, are decoupled from the outer systems.

For each $\I\in \Finf$, the conditional free energy of a state $\psi$ on $\Al$ 
 is defined  by 
\begin{equation}
\label{eq:FreeI}
\FIpsi:=\psi(\HIopen)-\frac{1}{\beta}\ScI(\psi).
\end{equation}

Using  the conditional free energy, we formulate the notion of  
local thermodynamical stability (LTS)  as follows.
\begin{definition}[LTS]
\label{df:LTS}
 A state $\vp$ of $\Al$  is said to satisfy the   
  local thermodynamical stability (LTS)
 with respect to the potential $\pot$ at inverse temperature $\beta>0$ if,  
for every  $\I\in \Finf$, 
\begin{equation}
\label{eq:LTSineq}
\FIvp\le \FIpsi
\end{equation}
 holds for all states  $\psi$ of $\Al$ satisfying the identity with $\vp$
on the complement subsystem on $\Ic${\rm{:}} 
\begin{equation}
\label{eq:FIamong}
\psiIc=\vpIc.
\end{equation}
\end{definition}

 The LTS condition requires  that  thermal equilibrium states are 
 characterized by  the minimality of the conditional free energy  
for \emph{each} local subsystem. These local subsystems  are 
 embedded in the total system $\Al$ and mutually interconnected. 

We note that the LTS condition itself   
 does not necessitate  a  $\cstar$-dynamics (time evolution) on $\Al$, 
 but it can be derived from  the KMS condition \cite{ASEW}.
Therefore, the LTS condition can be regarded as   
 a broader concept of thermal equilibrium.

\begin{remark}
\label{rem:LTS^solution}
Although the LTS condition is  formulated under  such general potentials, the actual existence of $\vp$ on $\Al$ satisfying the LTS condition  
 has been established only under more restrictive assumptions on $\pot$;
see \cite{SEWonly-LTStoKMSup, AM-LTS}. In this paper, we leave  aside this crucial problem and  implicitly assume the existence of such $\vp$.
\end{remark}


\subsection{Gibbs condition}
\label{subsec:GIBBS}
 We introduce  the Gibbs condition, another characterization  
 of thermal equilibrium for the quasi-local $\cstar$-system $\Al$.
It resembles local Gibbs states given in \eqref{eq:locGib-I}. 
However, it is intended for infinitely extended systems,  
 and its mathematical formulation uses   
 Tomita--Takesaki theory \cite{TAKE2book}.
We  briefly  recall  some necessary tools from Tomita--Takesaki theory.

\begin{definition}[Modular states]
\label{df:MODULAR}
Let $\vp$ be a state on $\Al$, and let  
 $\bigl(\Hilvp,\; \pivp,\; \Omevp  \bigr)$
 be its  GNS representation. 
Let $\vNMvp$ denote the von Neumann algebra 
 generated by this  representation, i.e., the weak closure of   
 $\pi_{\vp}(\Al)$ on $\Hilvp$.
If the GNS vector  $\Omevp$ is separating for   
 $\vNMvp$, then the state $\vp$ is called 
 a modular state. Let $\Delta_\vp$  and $\sigtvp$ $(t\in \R)$
 denote  the modular operator and  the
modular automorphism group, respectively,  
 related by  $\sigtvp={\rm{Ad}}(\Delta_\vp^{it})\in {\rm{Aut}}(\vNMvp)$ ($t\in \R$). The  weak extension of $\vp$  to 
 the  von Neumann algebra $\vNMvp$
 satisfies the KMS condition  with respect to the modular automorphism group
 at  inverse temperature $\beta=-1$, as in  Definition~\ref{df:KMS}.  
\end{definition}

The notions of perturbed dynamics and perturbed states for a modular state $\vp$  \cite{ARAKI-RH73} play crucial roles. For each self-adjoint element 
 $k=k^\ast \in \vNMvp$
the perturbed vector is given by
\begin{equation}
\label{eq:Omevpk}  
\Omevpk 
:=\exp\left\{\frac{1}{2}\left(\log \Delta_\vp+k\right)\right\}\Omevp\in 
V_{\vp}, 
\end{equation}
 where  $V_{\vp}$ denotes the natural positive cone in 
the GNS Hilbert space $\Hilvp$ associated with  the modular state  $\vp$.
Given a self-adjoint element  $h=h^{\ast}\in \Al$,  
 the perturbed positive linear  functional  $\vph$ on $\Al$  
is defined by 
\begin{equation}
\vph(A)\equiv
  \left(\Omevppih,\,\pivp(A) \Omevppih\right) \quad (A \in \Al).
\end{equation}
The  perturbed state on $\Al$ is obtained  by normalization as   
\begin{equation}
\label{eq:pert-state}
[\vph]:= \frac{1}{\vph(\idenA)}\vph.
\end{equation}
The perturbed  modular  automorphism group
$\sigtvph$ ($t\in \R$) is determined by 
 the following infinitesimal equality  
\begin{equation*}
\frac{d}{dt} \left(   
\sigtvph(x)-\sigtvp(x)
\right)_{t=0} = i  
\left[\pivp(h),\, x \right]
\end{equation*}
for every analytic element $x \in \vNMvp$ 
 with respect to $\sigtvp$ ($t\in \R$).
The perturbed state $[\vph]$ has its modular automorphism group $\sigtvph$ ($t\in \R$).

The Gibbs condition associated with $\pot$ relates  
a global state defined on $\Al$
 to the  local Gibbs states given in \eqref{eq:locGib-I}
 as follows.
\begin{definition}[Gibbs condition]
\label{df:Gibbs}
Suppose that a state $\vp$ of $\Al$
 is a modular state.
It satisfies  
the Gibbs condition with respect to $\pot$
 at $\beta$ if for each $\I\in \Finf$,  
 the   perturbed  state 
$[\vpbIsur]$   yields  the local 
Gibbs state $\rhopotI$ on $\AlI$
as given in \eqref{eq:locGib-I} when restricted to the subsystem $\AlI$.
\end{definition}

The Gibbs condition  further implies the product formula of 
the  perturbed states by surface energies.
\begin{pro}[\cite{ARAKI-KMS-VAR}, 
$\S$9.2 \cite{ENRICO}, $\S$7.5 of \cite{RMP-AM}]
\label{prop:GIBBS-product}
Let $\vp$ denote an arbitrary Gibbs  state 
 for  $\pot$ at $\beta$ for the quantum spin lattice system.
Then the  perturbed state by the surface energy  
has the following product formula{\rm{:}} 
\begin{equation}
\label{eq:Product-Gibbs}
[\vpbIsur]= \rhopotI \otimes 
 [\vpbIsur]\!\!\upright_{\Ic}. 
\end{equation}
For the fermion  lattice system,
 assume further that the Gibbs state $\vp$ is even.
Then 
\begin{equation}
\label{eq:FER-Product-Gibbs}
[\vpbIsur]= \rhopotI \gotimes 
 [\vpbIsur]\!\!\upright_{\Ic}.
\end{equation}
\end{pro}
Note that Gibbs states are not necessarily pure phases (factor states). 
The known relationship between the LTS condition (Definition~\ref{df:LTS}) 
and the Gibbs condition (Definition~\ref{df:Gibbs}) is as follows. 
\begin{pro}[\cite{ASEW, AM-LTS}]
\label{prop:GIBBS-LTS}
If a  state $\vp$ of the quantum spin lattice system
satisfies the Gibbs condition, then it satisfies the LTS condition.
If an even state $\vp$ of 
the  fermion lattice system satisfies the Gibbs condition, then it satisfies the LTS condition.
\end{pro}

\begin{remark}
\label{rem:KMS-Progyaku}
The converse implication of Proposition~\ref{prop:GIBBS-LTS} is  as follows.
If the potential 
 generates a  $\cstar$-dynamics on $\Al$, 
then the LTS condition implies  the KMS condition, 
which further yields  the Gibbs condition \cite{SEWonly-LTStoKMSup}.
 See Proposition~\ref{prop:KMS-GIBBS}.
\end{remark}

\begin{remark}
\label{rem:evenness}
If the state $\vp$ of  the fermion lattice system 
 satisfies both the LTS condition (Definition~\ref{df:LTS}) and the Gibbs condition (Definition~\ref{df:Gibbs}), then the evenness of $\vp$  follows, as shown in \cite{MORIYAgrading}.  
We conjecture that the evenness of $\vp$ can be derived from either of them  alone. (Note that the LTS condition in  Definition~\ref{df:LTS} corresponds to LTS-P, not LTS-M in \cite{AM-LTS}.)
\end{remark}

\subsection{KMS condition}
\label{subsec:KMS}
As we have noted before, the KMS condition 
is  not required  for our  thermal area law 
 which will be established in Section~\ref{sec:TAL}.  
Nonetheless,  we  will later use  certain properties of  
 the KMS condition in  Sections~\ref{sec:ONE}, \ref{sec:ONE-PROOF}. 
In fact, it is possible, and may even be natural, to start from  the KMS condition, since the KMS condition stands at the top of the hierarchy of thermal equilibrium conditions in quantum systems, implying other known conditions 
including the LTS condition and the Gibbs condition given in previous subsections, see \cite{BRA2}, \cite{RMP-AM}. 

We shall recall 
 the KMS condition in the present setting  of  quantum lattice systems.
Let $\delta_\pot$ denote the  derivation on $\Alloc$
associated with the potential $\pot$, defined  for every $\I\in \Finf$,   
\begin{equation}
\label{eq:delpot}
\delta_\pot(A):= i[\HIopen,\,A] \quad (A \in \AlI).
\end{equation}
Assume that
 $\delta_\pot$ generates 
 a $\cstar$-dynamics associated with $\pot$, that is, there 
 exists  a strongly continuous one-parameter group of $*$-automorphisms
  $\altpot:=\exp(it \delta_\pot)$ ($t\in\R$) of $\Al$.
\begin{definition}[KMS condition \cite{HHW, BRA2}]
\label{df:KMS}
A state $\vp$ of $\Al$ is called  an $(\altpot,\,\beta)$-KMS  state 
  if, for every $A, B \in \Al$,  there exists a complex-valued 
 function $F_{A, B}(z)$ of $z\in \C$  
  such that $F_{A, B}(z)$ is continuous and bounded on the closed strip
  $0 \le \operatorname{Im} z \le \beta$, holomorphic
 on its interior, and satisfies  
\begin{equation}
\label{eq:KMS}
F_{A, B}(t)=\vp \bigl(A \altpot(B) \bigr),\quad   
F_{A, B}(t+i \beta)=\vp \bigl(\altpot (B)A \bigr) \quad (t\in \R).
\end{equation}
\end{definition}

The following  result was mentioned  in Remark~\ref{rem:KMS-Progyaku}.
\begin{pro}[Theorem 9.1 in \cite{ENRICO}]
\label{prop:KMS-GIBBS}
Every  $(\altpot,\,\beta)$-KMS state 
$\vp$ is a modular state and  satisfies  
\begin{equation} 
\label{eq:KMSscale}
\sigtvp \bigl(\pivp(A)\bigr)
=\pivp\bigl(\almbtpot(A) \bigr),\quad  A \in \Al,
\end{equation} 
where $\sigtvp$ ($t\in \R$) denotes  the modular 
 automorphism group with respect to $\vp$
 in Definition~\ref{df:MODULAR}.
Moreover, $\vp$ satisfies  the  Gibbs condition with respect to $\pot$ at $\beta$ in Definition~\ref{df:Gibbs}.
\end{pro}

Take any $h=h^{\ast} \in \Al$.   
The  perturbation  of the $\cstar$-dynamics 
 $\altpot$ ($t\in\R$) by this self-adjoint element 
 is given by  the  $\cstar$-dynamics
  $\altpoth$ ($t\in\R$)  with  its generator  
\begin{equation}
\label{eq:INN0}
\delta_\pot^h(A)\equiv\delta_\pot(A)+i[h,\,A] \quad (A \in \Alloc).
\end{equation}
A state $\vp$ satisfies the $(\altpot,\,\beta)$-KMS condition 
if and only if the perturbed state  $[\vpmbh]$ satisfies the $(\altpoth,\,\beta)$-KMS condition. This establishes  a one-to-one correspondence
 between the set of $(\altpot,\,\beta)$-KMS states 
 and  the set  of $(\altpoth,\,\beta)$-KMS states.

\section{Thermal area law}
\label{sec:TAL}
In this section, we present the main result of this paper, the
thermal area law for quantum spin lattice systems and fermion lattice systems. 
As in Section~\ref{sec:EQUILIBRIUM-STATES}, 
let $\vp$ denote an arbitrary thermal equilibrium state, characterized by the LTS condition at inverse temperature $\beta$.

\subsection{$\cstar$-algebraic thermal area law and its proof}
\label{subsec:TAL-STATE}
In this subsection, we present 
the thermal area law in the $\cstar$-algebraic formulation 
 for both quantum spin lattice systems and fermion lattice  systems. 
For the notion of van Hove limit,  which is 
a  rigorous  formulation of the thermodynamic limit, 
we refer to  Section 6.2.4 of \cite{BRA2}.

\begin{thm}
\label{thm:MAIN}[Thermal area law for quantum spin lattice systems]
Consider the  quantum spin lattice system $\Al$.
Suppose that a state  $\vp$ of $\Al$ satisfies 
 the local thermodynamical stability (LTS)
  with respect to the potential $\pot$ at inverse temperature $\beta>0$.
Let $\rmA$ be an arbitrary finite region. For any (finite or infinite) 
region $\rmB$ outside $\rmA$, 
the following inequality for the mutual entropy holds{\rm{:}}
\begin{equation}
\label{eq:THERMALmain}
I_{\vp}(\rmA:\rmB)\le I_{\vp}(\rmA:\rmAc)
\le  
 \beta(\vp_{\rmA}\otimes
\vp_{\rmAc}-\vp)\left(H_{\partial \rmA}\right) \le 
2 \beta \lVert H_{\partial \rmA} \rVert.
\end{equation}
If the  surface energies per volume vanish in the van Hove limit as
\begin{align}
\label{eq:vHsurface}
{\rm{v.H.}}\lim_{\rmA \nearrow  \Gamma} 
\frac{\lVert H_{\partial \rmA} \rVert }{|\rmA|}
=0, 
\end{align}
then 
\begin{align}
\label{eq:vHtozero}
{\rm{v.H.}}\lim_{\rmA \nearrow  \Gamma} 
\frac{I_{\vp}(\rmA : \rmAc )}{|\rmA|}=0. 
\end{align}
\end{thm}

\begin{proof}
For  $\psi$ satisfying  the condition \eqref{eq:FIamong} in Definition~\ref{df:LTS} of LTS, we  now  take
the  product state made by 
the reduced states of $\vp$ to $\rmA$ and the complement $\rmAc$: 
\begin{equation}
\label{eq:trial-product}
\vp_{\rmA}\otimes \vp_{\rmAc}.
\end{equation}
 Then by plugging this product state into the inequality \eqref{eq:LTSineq} of the LTS condition, we obtain 
\begin{equation}
\label{eq:TOKU}
\FA(\vp) \le \FA(\vp_{\rmA}\otimes \vp_{\rmAc}).
\end{equation}
By recalling the formula of the conditional free energy \eqref{eq:FreeI},  
 the inequality \eqref{eq:TOKU} yields  
\begin{align}
\label{eq:ent-ene-1}
\ScA(\vp_{\rmA}\otimes \vp_{\rmAc})
-\ScA(\vp)
\le 
\beta\bigl(\vp_{\rmA}\otimes \vp_{\rmAc}
-\vp
\bigr)(\HAopen).
\end{align}
We consider the entropy term in the left-hand side of  \eqref{eq:ent-ene-1}.  
By the  additivity of von Neumann entropy for  product states, 
 we have 
\begin{align}
\label{eq:entc-PRODUCT}
\ScA(\vp_{\rmA}\otimes \vp_{\rmAc})
=S_{\rmA}(\vp).
\end{align}
Thus, the left-hand side of \eqref{eq:ent-ene-1} is equal to
$S_{\rmA}(\vp)
-\ScA(\vp)=I_{\vp}(\rmA:\rmAc)$ by \eqref{eq:MUTUAL-IandIc}.
Next,  we  consider the energy term in the right-hand side
of \eqref{eq:ent-ene-1}. 
\begin{align}
\label{eq:energy-openest}
\bigl(\vp_{\rmA}\otimes \vp_{\rmAc}
-\vp
\bigr)(\HAopen)&=
\bigl(\vp_{\rmA}\otimes \vp_{\rmAc}
-\vp
\bigr)(\HA+\HAsur)\nonum\\
&=
\bigl(\vp_{\rmA}\otimes \vp_{\rmAc}
-\vp
\bigr)(\HA)+
\bigl(\vp_{\rmA}\otimes \vp_{\rmAc}
-\vp
\bigr)(\HAsur)
\nonum\\
&=
0+\bigl(\vp_{\rmA}\otimes \vp_{\rmAc}
-\vp
\bigr)(\HAsur).
\end{align}
Thus,  \eqref{eq:ent-ene-1}
yields
 \begin{equation}
\label{eq:main-mae}
I_{\vp}(\rmA:\rmAc)\le  
\beta\left(\vp_{\rmA}\otimes\vp_{\rmAc}-\vp\right)\left(H_{\partial \rmA}\right). 
\end{equation} 
Using this together with  
the inequality  $I_{\vp}(\rmA:\rmB)\le 
I_{\vp}(\rmA:\rmAc)$ and 
the obvious inequality $\lVert  \vp_{\rmA}\otimes\vp_{\rmAc}-\vp \rVert\le 2$,   we obtain \eqref{eq:THERMALmain}.

If \eqref{eq:vHsurface} is satisfied, then 
the inequality \eqref{eq:THERMALmain}
 shown above implies
\eqref{eq:vHtozero}.
\end{proof}

\begin{remark}
\label{rem:vanenter}
Theorem~\ref{thm:MAIN} is  
analogous to  the main result in \cite{vanENTER}, 
which establishes the equivalence between the mean von Neumann entropy and the mean  conditional entropy for translation-invariant thermal equilibrium states. 
Theorem~\ref{thm:MAIN} instead emphasizes the state correlations captured by the mutual entropy.
\end{remark}

We derive a similar statement 
to Theorem~\ref{thm:MAIN} for the  fermion lattice system 
with  some modifications. 
\begin{thm}[Thermal area law for fermion lattice systems]
\label{thm:FER-MAIN}
Suppose that  a state  $\vp$ of the  fermion lattice system $\Al$ 
 satisfies the   local thermodynamical stability (LTS) 
with respect to the potential $\pot$ at inverse temperature $\beta>0$.
Assume further that $\vp$ is an even state. Let $\rmA$ be an arbitrary finite region. For any (finite or infinite) 
region $\rmB$ outside $\rmA$,  the following estimate holds{\rm{:}}
\begin{equation}
\label{eq:FER-THERMALmain}
I_{\vp}(\rmA:\rmB)\le I_{\vp}(\rmA:\rmAc)
\le  
 \beta(\vp_{\rmA}\gotimes
\vp_{\rmAc}-\vp)\left(H_{\partial \rmA}\right) \le 
2 \beta \lVert H_{\partial \rmA} \rVert.
\end{equation}
\end{thm}

\begin{proof}
As in \eqref{eq:trial-product}, we take the product-state extension of 
the reduced states of $\vp$ to the finite region $\rmA$ and its  complement region $\rmAc$ following \cite{AM2003EXT}
\begin{equation}
\label{eq:trailai-gpro}
\vp_{\rmA}\gotimes \vp_{\rmAc}.
\end{equation}
Then by plugging this even product state  
 into the inequality \eqref{eq:LTSineq} of the LTS condition, we
 obtain an  analogous  estimate to that in \eqref{eq:ent-ene-1}  
 replacing $\otimes$ by $\gotimes$. Since any product state of the fermion system implies the additivity of von Neumann entropy,  
(in fact, the converse  also holds 
\cite{MARKOV}), we have   
\begin{align}
\label{eq:g-entc-PRODUCT}
\ScA(\vp_{\rmA}\gotimes \vp_{\rmAc})=S_{\rmA}(\vp).
\end{align}
A similar derivation as in  \eqref{eq:energy-openest}
 holds for the fermion lattice system due to the evenness of the states 
and the local Hamiltonians.
Thus  we obtain  an analogous inequality to that of
 \eqref{eq:main-mae}
 which immediately implies the  asserted estimate \eqref{eq:FER-THERMALmain}
 for the fermion lattice system.
\end{proof}

A common expression of the thermal area law as in \eqref{eq:HpartialA}  
 can be derived straightforwardly in the $\cstar$-algebraic setting  as follows.\begin{cor}
\label{cor:GEOM-thermal}
Consider any state $\vp$ satisfying the LTS condition 
as in  Theorem~\ref{thm:MAIN} for the quantum spin lattice system, 
 or any even state satisfying the LTS condition
as in Theorem~\ref{thm:FER-MAIN} for the fermion lattice system.
Suppose that there exists  some constant 
 $c_{\pot}>0$ such that the estimate
\begin{align}
\label{eq:BOUNDED-Area}
\lVert H_{\partial \rmA} \rVert\le c_{\pot} |\partial \rmA|
\end{align}
 holds. Then, 
for any $\rmB$ outside $\rmA$, 
\begin{align}
\label{eq:areaCONST}
I_{\vp}(\rmA : \rmB) \le c_{\pot}|\partial \rmA|.
\end{align}
\end{cor}

\begin{remark}
\label{rem:cor-geomet}
The  assumption  \eqref{eq:BOUNDED-Area} 
 of Corollary~\ref{cor:GEOM-thermal}
holds if the potential $\pot$
 is of  finite range. When  $\pot$ has infinite range, 
the support of the surface energy $H_{\partial \rmA}\in \Al$
 is not strictly local in the  $\cstar$-algebra. 
In such cases, a geometrical interpretation of  ${\partial \rmA}$ 
 in \eqref{eq:HpartialA} in terms of $\pot$ becomes necessary,
by introducing  an appropriate notion of ``almost local.''
\end{remark}

\subsection{Correlation estimates}
\label{subsec:TAL-COR}
We recall the Pinsker inequality for the  quantum relative entropy
 \cite{CIS}. For two states $\psi$ and $\omega$  
\begin{align}
\label{eq:pinsker-general}
\lVert \psi- \omega \rVert^2   \le 2S(\psi \mid \omega)
\end{align}
The Pinsker inequality has been extended to  Araki's quantum relative entropy, 
 as shown in Theorem 3.1 in \cite{HOT}; see Theorem 5.5 of \cite{OHYA-PETZ}.

From the thermal area law  shown in Theorem~\ref{thm:MAIN}
 and  Theorem~\ref{thm:FER-MAIN}, we can derive an estimate 
 between the given thermal equilibrium state $\vp$ and 
the product state $\vp_{\rmA}\otimes\vp_{\rmAc}$  
 using the Pinsker inequality, by the same reasoning as in the finite-dimensional case \cite{WOLF}. 

\begin{cor}
\label{cor:cluster}
For any state $\vp$ of the quantum spin lattice system that satisfies the area law as in \eqref{eq:THERMALmain}, the following estimate holds
\begin{align}
\label{eq:PINSKER}
\lVert \vp_{\rmA}\otimes
\vp_{\rmAc}-\vp \rVert^2 
\le 4
\beta \lVert H_{\partial \rmA} \rVert.
\end{align}
For any even  state $\vp$ of the fermion lattice  system that   
  satisfies the area law as in \eqref{eq:FER-THERMALmain}, 
  the following  estimate holds
\begin{align}
\label{eq:FER-PINSKER}
\lVert \vp_{\rmA}\gotimes
\vp_{\rmAc}-\vp \rVert^2 
\le 4 
\beta \lVert H_{\partial \rmA} \rVert.
\end{align}
In particular, for both 
the quantum spin lattice system and 
 the fermion lattice  system,  the  estimate
\begin{align}
\label{eq:CLUSTER}
\Bigl| \vp (\OA\OB)-\vp (\OA)\vp (\OB) \Bigr| \le 2  
\left( 
\beta \lVert H_{\partial \rmA} \rVert
\right)^{\frac{1}{2}}
\end{align}
holds for any $\OA\in \AlA$, $\OB\in \AlAc$ such that $\lVert \OA \rVert\le 1 $
 and $\lVert \OB \rVert\le 1$.
\end{cor}

\begin{remark}
\label{rem:CLUSTER-limitation}
The universal bound on spatial correlations derived from  the mutual entropy estimate is rather coarse, 
as pointed out in some  physics literature such as \cite{Bernigau}.
This inherent limitation of mutual entropy becomes more evident in infinitely extended systems.  
Consider any potential $\pot$ that exhibits multiple equilibrium states, possibly due to  spontaneous symmetry breaking. The thermal area law  
as in Theorems~\ref{thm:MAIN} and \ref{thm:FER-MAIN} 
 is valid  for all thermal equilibrium phases, as well as  any  statistical mixture of them, which gives rise to a non-factor von Neumann algebra by GNS construction.   On the other hand, any non-factor state of quasi-local $\cstar$-systems does not satisfy the spatial cluster property.  
Consequently,  the thermal area law itself does not exclude states 
without the spatial cluster property. 
The above observation  based on the underlying quasi-local $\cstar$-systems
seems  difficult to capture  by the conventional box procedure, since  any non-factor thermal equilibrium state lacks a definite value for certain order parameters, and thereby induces \emph{effective} 
long-range interactions with unstable  surface energies \cite{NAR}, 
even when the given potential $\pot$ is of finite-range.
\end{remark}

\begin{remark}
\label{rem:CLUSTER-II}
This remark complements  Remark~\ref{rem:CLUSTER-limitation} above.
 When a thermal equilibrium state exhibits strong spatial decay, 
certain refinements of  the thermal area law may imply 
stronger independence between disjoint regions. For examples of such estimates, see \cite{Bernigau, BLUHM}.
\end{remark}

\subsection{Area law for ground states in terms of quantum mutual entropy}
\label{subsec:GR}
The thermal area law formulated in \cite{WOLF} is a natural extension  of the area law for ground states 
(zero-temperature equilibrium states)  \cite{HAS2007} to thermal states.  
 This correspondence is evident from the identity 
$I_\rho(\rmA:\rmAc)=2S_{\rmA}(\rho)$ for any pure state $\rho$
 on a finite-dimensional tensor-product quantum system. 

For infinitely extended  quantum lattice systems as well,  
the area law for ground states is  defined  
by the uniform boundedness of von Neumann entropy (entanglement entropy). 
In \cite{MAT-BOUND} \cite{UK}, its precise formulation and the conditions 
under which it is  satisfied  have been studied.
From the finite-dimensional case, 
one may naturally conjecture  that the  area law for ground states 
can  be formulated in terms of the mutual entropy 
instead of the von Neumann entropy.

As in previous research on ground states, we may restrict 
 the subregions to be considered.
Let $\BOX$ denote a set of (sufficiently many) finite subsets of $\Finf$ 
that eventually cover the whole lattice $\Gamma$.  
For concreteness, we may take $\BOX$ to be  
the collection of box regions containing the origin.
We can derive the following one-sided implication.
\begin{pro}[Area law formula for ground states in terms of mutual entropy]
\label{prop:AREA-GR}
Let $\rho$ be a pure state on the quantum spin lattice system,  
or a pure even state on the fermion lattice system.
If it satisfies 
 the uniform boundedness of the von Neumann entropy{\rm{:}}
\begin{align}
\label{eq:BOU-vN}
S_{\rmA}(\rho)\le c |\partial \rmA|
\end{align}
for all $\rmA \in \BOX$ with  some uniform constant $c>0$, 
then 
\begin{align}
\label{eq:BOU-MUT}
I_{\vp}(\rmA : \rmAc) \le 2c |\partial \rmA|
\end{align}
for all $\rmA \in \BOX$. 
\end{pro}

\begin{proof}
By \eqref{eq:MUT-2von}, 
the assumption \eqref{eq:BOU-vN} readily implies \eqref{eq:BOU-MUT}.
\end{proof}

\begin{remark}
\label{rem:GR}
While the thermal area law holds universally, 
  the area law for ground states is not always satisfied; 
 see e.g. \cite{AREA-rev}, \cite{WOLF-06}. 
Its validity  has been an important  issue 
in condensed matter physics and mathematical physics.  
(Proposition~\ref{prop:AREA-GR}  does not address  this question.)   
\end{remark}

\section{Mutual entropy between disjoint infinite regions}
\label{sec:ONE}
We continue to investigate the mutual entropy $I_{\vp}(\rmA:\rmB)$
 for thermal equilibrium states $\vp$, but now  
in the situation where both regions $\rmA$ and $\rmB$ are infinite.
In this case, the identities $I_\vp(\rmA:\rmB)=S_{\rmA}(\vp)+S_{\rmB}(\vp)-S_{\rmAB}(\vp)$ in \eqref{eq:MUT} and  $I_\vp(\rmA:\rmB)=S_{\rmA}(\vp)-\ScAB(\vp)$
in \eqref{eq:MUT-CondExp-finite} are  generically invalid, since
  the local von Neumann entropies may diverge. 
Nevertheless, if $\vp$ exhibits  sufficient independence between $\rmA$ and $\rmB$, then   $I_{\vp}(\rmA:\rmB)$ can remain finite; an obvious example is product states between $\rmA$ and $\rmB$. 
We shall establish this finiteness for all finite-range 
translation-invariant models on one-dimensional quantum (spin and fermion) 
lattice systems.

\begin{remark}
\label{rem:AQFT-FINITE}
In (algebraic) quantum field theory, the finiteness of the mutual entropy 
 of vacuum states between disjoint subregions
 has been verified  in various settings; see e.g. \cite{CAS-HUE-09}, \cite{LONGO-CFT}, \cite{XU}.
\end{remark}

\subsection{One-dimensional lattice systems: setup and notation}
\label{subsec:ONE-setup}
In this section, we focus on the quantum spin system and the fermion system on the one-dimensional integer lattice $\Z$. To make the one-dimensional lattice explicit, 
 we denote the  total $\cstar$-system  by $\AlZ$, 
instead of the general  notation $\Al$ used so far.  
Similarly, we write $\AlZloc$ 
 for the local algebra, and 
$\AlZI$  for the subsystem  on $\I\subset \Z$.

We divide the total space $\Z$  into the disjoint regions $\ZL$ and $\ZR$, 
 defined as  
\begin{equation*}
\ZL \equiv \NN_{-}:=\{ \cdots,  -5, -4, -3, -2, -1 \} \subset \Z,
\end{equation*}
and 
\begin{equation*}
\ZR \equiv \{0\} \cup \NN_{+}:=\{0, 1, 2, 3, 4, 5, \cdots \} \subset \Z.
\end{equation*}
We take the left-sided region 
$\ZL$ and  the right-sided region $\ZR$
for the pair of disjoint regions $\rmA$ and $\rmB$. 

 We denote the quasi-local $\cstar$-system on $\ZL$ by $\AlL$, which is 
  identical to $\AlZL$ including its  quasi-local structure.
 We denote the  quasi-local $\cstar$-system  on $\ZR$ by  $\AlR$, which is 
  identical to $\AlZR$ including its  quasi-local structure.
When they denote fermion lattice systems, 
 the fermion grading automorphisms 
 $\THEL$  on $\AlL$ and  $\THER$  on $\AlR$ are given 
 as in \eqref{eq:CAR-THETA}.
By definition,  $\AlL$ and  $\AlR$ are distinct  $\cstar$-systems.
In practice, however, we will sometimes identify $\AlL=\AlZL$ and  $\AlR=\AlZR$
 when there is no risk of confusion.
Let  $\FinfL$ and $\FinfR$ denote the sets of all 
 finite subsets of $\ZL$ and $\ZR$, respectively.
Let $\AlLloc:=\bigcup_{\I \in\FinfL}\AlI$ 
 and  $\AlRloc:=\bigcup_{\I \in\FinfR}\AlI$; they are the
 local algebras of  $\AlL$ and $\AlR$, respectively.

We impose assumptions on the potential $\pot$ on $\AlZ$.
First,  $\pot$ is translation invariant.
 Let  $\{\tau_x\in {\rm{Aut}}(\AlZ),\;x\in \Z\}$ denote the shift-translation automorphism group on $\AlZ$. For each $\K\in \Finf$, 
\begin{equation}
\label{eq:POT-SHIFT}
\tau_x(\potK)=\potKx \in \AlKx \quad \forall x\in\Z. 
\end{equation}
Second,  $\pot$ is of finite-range. 
For each $\I\in \Finf$, let $d(\I)$ denote the largest distance 
 between two points of $\I$.
Let $d(\pot)$ denote the supremum of all $d(\I)$ such that $\potI$ is nonzero.
We assume  $d(\pot)<\infty$. 
Thus, within this and the next section, 
 $\pot$ is a translation-invariant finite-range 
  potential on $\AlZ$.

Owing to  the assumption $d(\pot)<\infty$, 
the surface energy between $\ZL$ and $\ZR$ is well defined  as  
\begin{equation}
\label{eq:WLR}
\WLR:= \sum_{\K: \; \K\cap \ZL \ne \emptyset, \K\cap\ZR \ne \emptyset}
\potK \in \AlZloc.
\end{equation}
In the notation used in  \eqref{eq:HIsur}, 
$\WLR$ would be denoted as either $\HLsur$ or $\HRsur$. 
However, since $\ZL$ and $\ZR$ play symmetric roles, we adopt the notation $\WLR$ to explicitly express the dependence on both regions. 

Our assumption on the potential $\pot$ is 
stronger than necessary,  chosen mainly for technical convenience.
We shall mention this  point in Remark~\ref{rem:POT-ASSUMPTION}  
after presenting  the proof.

\subsection{Finite mutual entropy between $\ZL$ and $\ZR$ for thermal equilibrium states}
\label{subsec:shoiwngFIN}
Given any translation-invariant finite-range potential
 $\pot$ on $\AlZ$ and any $\beta>0$, 
let $\vp$ denote the  thermal equilibrium state 
with respect to $\pot$ at  inverse temperature $\beta$.
The uniqueness of such $\vp$  for the one-dimensional 
quantum spin lattice system 
follows from \cite{ARA69, ARA75unique}, and the proof 
remains valid  for the one-dimensional fermion lattice system \cite{RMP-AM}.
This $\vp$ automatically satisfies all of the LTS, Gibbs, and KMS conditions; 
see  \cite{BRA2}, and also \cite{AM-LTS, RMP-AM}.

In Theorem~\ref{thm:ONE-FINITE}, we   
establish the finiteness of the mutual entropy 
$I_{\vp}(\ZL:\ZR)$ between the disjoint infinite regions $\ZL$ and $\ZR$. 
 This result can be  regarded as 
a natural  extension of  the thermal area law as 
in Theorems~\ref{thm:MAIN} and \ref{thm:FER-MAIN}.

\begin{thm}[Finite mutual entropy between $\ZL$ and $\ZR$]
\label{thm:ONE-FINITE}
 Let $\pot$ be any translation-invariant finite-range 
  potential on the  one-dimensional quantum spin or fermion 
lattice  system $\AlZ$. Let $\vp$ be the unique thermal equilibrium state 
 with respect to $\pot$ at inverse temperature $\beta>0$. 
Then the mutual entropy $I_{\vp}(\ZL:\ZR)$ of $\vp$ between the  
left-sided region $\ZL$ and the  right-sided region $\ZR$ is  finite, 
 and satisfies  the bound
\begin{equation} 
\label{eq:thmGOAL}
I_\vp(\ZL:\ZR)
\le 2 \beta \lVert \WLR \rVert. 
\end{equation} 
\end{thm}

To clarify the meaning of Theorem~\ref{thm:ONE-FINITE}, consider 
  a general state $\ome$ of $\AlZ$. 
If the mutual entropy $I_{\ome}(\ZL:\ZR)$ of $\ome$ is finite (or even small), 
 then  $\ome$ is close to the product state 
$\omeL\otimes\omeR$ formed from its reduced states. Let us recall the split property for states on $\AlZ$ between $\AlL$ and $\AlR$. This property requires the (quasi-)equivalence of the two states $\ome$ and $\omeL\otimes\omeR$ \cite{SPLIT2001}. It was noted in \cite{SPLIT2001} 
that the  thermal equilibrium state $\vp$ with respect to 
a translation-invariant finite-range potential $\pot$  on the one-dimensional quantum spin lattice system satisfies the split property, owing to the half-sided uniform spatial cluster property \cite{ARA69}.

The proof of Theorem~\ref{thm:ONE-FINITE} (the finiteness of  
$I_{\vp}(\ZL:\ZR)$) is postponed to Section~\ref{sec:ONE-PROOF}.
Instead, in this section, we  shall address two  notable consequences of
Theorem~\ref{thm:ONE-FINITE}. 
The first one is about the quantum entanglement between $\AlL$ and $\AlR$.
\begin{cor}[Finite quantum entanglement between $\ZL$ and $\ZR$]
\label{cor:ENTANG-finite}
The relative-entropy entanglement between $\AlL$ and $\AlR$ of the thermal equilibrium state $\vp$ on $\AlZ$ is defined as 
\begin{equation} 
\label{eq:RELAentanglement}
E_{{RE}}(\vp)(\ZL:\ZR)
:=\inf\{S(\vp \mid \omega ): \omega \in \SEP_{\ZL:\ZR}\},
\end{equation} 
where $\SEP_{\ZL:\ZR}$ denotes the set of separable states on $\AlZ$  
with respect to $\AlL$ and $\AlR$. Here the subscript 'RE' indicates measurement via relative entropy. Under the same assumptions as in Theorem~\ref{thm:ONE-FINITE}, $E_{{RE}}(\vp)(\ZL:\ZR)$ is finite.
\end{cor}
\begin{proof}
Since the relative-entropy entanglement (commonly called  ``relative entropy of entanglement'' \cite{Vedral97}) is bounded above by the mutual entropy, 
the finiteness of $E_{{RE}}(\vp)(\ZL:\ZR)$ immediately follows from Theorem~\ref{thm:ONE-FINITE}. Note that the definition of the relative-entropy entanglement in  the  general von Neumann algebra setting can be found in
  Definition~11 of \cite{HO-SA}. By employing the notion of separable states on fermion lattice systems presented in \cite{MOR-SEP}, the  argument used for the quantum-spin lattice system  applies to the fermion lattice system.
\end{proof}

The following corollary is another direct consequence of 
 Theorem~\ref{thm:ONE-FINITE}. It demonstrates a
 remarkable destruction of quantum entanglement between $\ZL$ and $\ZR$  
induced by any (even slight) positive temperature.  
In this corollary, 
we explicitly write the $\beta$-dependence of equilibrium  states.
\begin{cor}[Thermal destruction of quantum entanglement]
\label{cor:KIERU}
 Let $\pot$ be any translation-invariant finite-range potential 
on the one-dimensional quantum spin or fermion 
lattice system $\AlZ$ as in Theorem~\ref{thm:ONE-FINITE}. 
Let $\vpg$ be any pure ground state with respect to $\pot$. 
Let $\vpbeta$ denote the unique thermal equilibrium state 
with respect to  the same $\pot$ at inverse temperature $\beta>0$.
Suppose that $\vpg$ does not satisfy the split property 
between $\AlL$ and $\AlR$.  
Then $I_{\vpg}(\ZL:\ZR)=\infty$ whereas $I_{\vpbeta}(\ZL:\ZR)<\infty$ 
for all $\beta>0$. 
\end{cor}
\begin{proof}
 Since the finiteness condition $I_{\ome}(\ZL:\ZR)<\infty$ implies that 
 $\omeL\otimes\omeR$ quasi-contains $\ome$ in the GNS construction
 according to  Lemma~2 of \cite{ARA75unique}, the violation of the split property between $\AlL$ and $\AlR$ 
of $\vpg$  implies $I_{\vpg}(\ZL:\ZR)=\infty$. 
On the other hand,  
$I_{\vpbeta}(\ZL:\ZR)$ remains finite  
for all $\beta>0$ by Theorem~\ref{thm:ONE-FINITE}. 
This proves the assertion.
\end{proof}

Quantum lattice models on $\Z$  that violate   
the split property between $\AlL$ and $\AlR$ 
 are often regarded as  critical  models of  conformal field theory (CFT).
 For rigorous characterizations and explicit examples 
 of finite-range potentials $\pot$ on $\AlZ$ 
that give rise to non-split ground states $\vpg$ on $\AlZ$, we refer to
 \cite{SPLIT2001} and  \cite{KMSW06}. 

\section{Proof of finite mutual entropy between $\ZL$ and $\ZR$}
\label{sec:ONE-PROOF}
In this section, 
we present the proof of Theorem~\ref{thm:ONE-FINITE}
 stated in the preceding section.  
Specifically, we establish the finiteness of  
\begin{align}
\label{eq:MUGEN-LR}
I_\vp(\ZL:\ZR)\equiv S(\vp \mid \vpL \otimes \vpR)
\end{align}
for the quantum spin lattice system on $\Z$, and    
\begin{align}
\label{eq:FER-MUGEN-LR}
I_\vp(\ZL:\ZR)\equiv S(\vp \mid \vpL \gotimes \vpR)
\end{align}
for the fermion lattice system on $\Z$.

Before proceeding, 
we  note that both formulas are well defined. 
Since $\vp$ is a KMS state, it is a faithful state on $\AlZ$.
 Consequently, both $\vpL \otimes \vpR$ 
 and  $\vpL \gotimes \vpR$ are faithful states as well, and hence 
 Araki's relative entropy expressions in \eqref{eq:MUGEN-LR} and
 \eqref{eq:FER-MUGEN-LR} are well defined.

The proof is divided into several steps.
We provide a number of structural results in Subsections~\ref{subsec:AG-LR}, \ref{subsec:PRO-EXT}, and \ref{subsec:DONALD}. Each subsection is given an informative title, as these results are formulated in a way that suggests interest beyond the present proof. 
With these preparations, we complete  
 the proof in Subsection~\ref{subsec:complete}. 
The argument is developed in parallel for the quantum spin and fermion cases, although the fermion case requires certain nontrivial modifications, which we explain in detail.

\subsection{Araki-Gibbs condition between $\ZL$ and $\ZR$}
\label{subsec:AG-LR}
Essentially, we aim to derive a certain independence (a product-like property)  
of the thermal equilibrium state $\vp$ on $\AlZ$
 between the half-sided subsystems $\AlL$ and $\AlR$. 
To this end, we introduce models
 on the separated  systems $\AlL$ and  $\AlR$   
 from the given  finite-range potential $\pot$ on $\AlZ$.

Let $\potL$ be the finite-range potential on $\AlL$ defined by   
\begin{equation}
\label{eq:POTLdef}
\potLK:=\potK\in \AlK,\quad \forall \K \in \FinfL. 
\end{equation}
Similarly, let $\potR$ be the finite-range potential on $\AlR$ defined by 
\begin{equation}
\label{eq:POTRdef}
\potRK:=\potK\in \AlK, \quad \forall \K \in \FinfR.
\end{equation}
Let $\delta_{\potL}$  and $\delta_{\potR}$ be the derivations  
associated with the potentials $\potL$  and $\potR$, respectively, as in \eqref{eq:delpot}.
Define 
$\altpotL:=\exp(it \delta_{\potL})$ ($t\in\R$), the  
$\cstar$-dynamics of $\AlL$ generated by the derivation 
 $\delta_{\potL}$ on  $\AlLloc$. Similarly, define
$\altpotR:=\exp(it \delta_{\potR})$ ($t\in\R$), the  
$\cstar$-dynamics of $\AlR$ generated by the derivation 
 $\delta_{\potR}$ on $\AlRloc$.  
The existence of $\altpotL$ and 
 $\altpotR$ follows from  the finite-range  of $\potL$ and $\potR$. 

By the main result of \cite{ARA75unique, KIS76unique}, there exists a unique $(\altpotL,\,\beta)$-KMS state on $\AlL$, denoted by $\vppotL$. Similarly, $\vppotR$ denotes the unique $(\altpotR,\,\beta)$-KMS state on $\AlR$.

The following proposition establishes a realization of the 
Araki-Gibbs condition 
in the present setting, where $\Z$ is split into $\ZL$ and $\ZR$. 

\begin{pro}[Araki-Gibbs condition between $\ZL$ and $\ZR$]
\label{prop:AG-PROD-LR}
Let $\vp$ denote the unique thermal equilibrium  state of the one-dimensional quantum spin lattice or fermion lattice system $\AlZ$ with respect to the translation-invariant finite-range  potential $\pot$ at $\beta>0$. 
For the quantum spin system on $\Z$, the following product 
 formula holds{\rm{:}} 
\begin{equation}
\label{eq:PERTW-ten}
[\vpbW]=\vppotL \otimes \vppotR.
\end{equation}
For the fermion lattice system on $\Z$, 
the following product formula holds{\rm{:}} 
\begin{equation}
\label{eq:PERTW-FER}
[\vpbW]=\vppotL \gotimes \vppotR.
\end{equation}
\end{pro}

\begin{proof}
First, we  verify the product formula for the perturbed dynamics of $\altpotmW$.
For the quantum spin system on $\Z$, 
\begin{equation}
\label{eq:PRO-DYN-QS}
\altpotmW=\altpotL \otimes \altpotR \in {\rm{Aut}}(\AlZ)
\quad (t\in \R),
\end{equation}
and for the fermion lattice system on $\Z$, 
\begin{equation}
\label{eq:PRO-DYN-FER}
\altpotmW=\altpotL \gotimes \altpotR \in {\rm{Aut}}(\AlZ)
\quad (t\in \R).
\end{equation}
We readily see that the above equalities as $\cstar$-dynamics on $\AlZ$ hold, since  
 the infinitesimal generators of 
$\altpotmW$ and $\altpotL \otimes \altpotR$ (resp. $\altpotL \gotimes \altpotR$)
 are both associated with the same (decoupled) potential $\potLR$ on $\AlZ$ defined by 
\begin{align}
\label{eq:POTLRdef}
\potLR (\K)&=\potK\in \AlK \quad \text{if}\ \K \in \FinfL\ \text{or} \ \K \in \FinfR,\nonum\\
\potLR (\K)&=0 \quad \text{otherwise}.
\end{align}
Namely, $\potLR$ is obtained from $\pot$ 
 by removing  all interactions  between $\ZL$ and $\ZR$.
Note that  no distinction arises in the fermion system in the above argument 
due to the evenness of the potential $\pot$. 

Since $\vppotL$ is the (unique) $(\altpotL,\,\beta)$-KMS state on  
$\AlL$, and  $\vppotR$ is the (unique) $(\altpotR,\,\beta)$-KMS state on $\AlR$, 
the product state $\vppotL \otimes \vppotR$ 
(resp. $\vppotL \gotimes \vppotR$) gives a  KMS state with respect to 
$\altpotL \otimes \altpotR$ (resp. $\altpotL \gotimes \altpotR$) at inverse temperature 
$\beta$ by Proposition~\ref{prop:EXT-AUT}.

Since $\vp$ is the unique $(\altpot,\,\beta)$-KMS state on $\AlZ$, 
its  perturbed state  $[\vpbW]$ corresponds to  the unique   
$\left(\altpotmW,\,\beta \right)$-KMS state on $\AlZ$ by the fundamental result on the perturbation of $\cstar$-dynamics and KMS states stated in  Subsection~\ref{subsec:KMS}. Thus, by the uniqueness of the KMS state with respect to 
 the same $\cstar$-dynamics, the product state 
$\vppotL \otimes \vppotR$ (resp. $\vppotL \gotimes \vppotR$)
 coincides with the perturbed KMS state $[\vpbW]$.
\end{proof}

\begin{remark}
\label{rem:Araki-Gibbs}
We shall state  some reflections on the Araki-Gibbs condition, which plays a pivotal role in this paper.
The term  ``Araki-Gibbs condition"  used in \cite{BRA2} 
does not actually stand for  a joint work  between  Huzihiro Araki and Josiah Willard Gibbs, unfortunately. 
Although the Araki-Gibbs condition appears to be akin to 
 the Dobrushin-Lanford-Ruelle (DLR) condition
 characterizing Gibbs measures in classical systems \cite{ARA-ION}, 
according to Araki, it was devised as an intermediate notion  
relating  the KMS condition to the variational principle.
 Among the consequences derived from the KMS condition, 
one  example is a  no-go theorem for quantum time crystals 
in thermal equilibrium, 
which was  presented in \cite{ARA68}, long before the proposal of quantum time crystals. As a related issue,  we  shall mention another work of Araki 
\cite{ARA64}, which forbids not only temporal (obviously) but also 
spatial (rather non-trivial) crystalline order for vacuum states in QFT; 
see \cite{MORIYA24}.
\end{remark}

\subsection{Product extension of states and automorphisms on disjoint regions}
\label{subsec:PRO-EXT}
In this subsection,  we  provide  some general 
 results on product extensions of automorphisms and 
states in disjoint subsystems, both for the quantum spin lattice system
and  for the fermion lattice system.
 These structural results are in fact  valid for  
 general boson and  fermion quasi-local $\cstar$-systems.
\begin{pro}[Product extension of automorphisms]
\label{prop:EXT-AUT}
Let $\alL$ denote a  $*$-automorphism of $\AlL$, and let $\alR$ denote a  $*$-automorphism  of $\AlR$.
For the quantum spin lattice system, 
there  exists a  product extension of $\alL$ and $\alR$ as
 a $*$-automorphism on $\AlZ${\rm{:}} 
\begin{equation} 
\label{eq:alpha-ext-tensor}
\alL\otimes \alR \in {\rm{Aut}}(\AlZ).
\end{equation} 
For the fermion lattice system, assume that each of   
 $\alL$ and  $\alR$ preserves the fermion grading on its respective system,    
\begin{align} 
\label{eq:alLR-grad-PRESERVE}
\alL \THEL=  \THEL \alL, \quad \alR \THER=  \THER \alR.
\end{align} 
Then, there  exists a
product extension of $\alL$ and $\alR$ as a $*$-automorphism on $\AlZ${\rm{:}}
\begin{equation} 
\label{eq:alpha-ext-FER}
\alL \gotimes \alR \in {\rm{Aut}}(\AlZ) 
\end{equation} 
such that 
\begin{align} 
\label{eq:g-AUT-action}
\alL \gotimes \alR\left(\sum_{k} A_k B_k\right)=  \sum_{k} \alL(A_k) \alR(B_k)
\end{align} 
for any finite sum $\sum_{k} A_k B_k \in \AlZ $ with  $A_k\in \AlL$ and
 $B_k\in \AlR$.
\end{pro}

\begin{proof}
For the quantum spin lattice system, 
 the  total system $\AlZ$ is given as 
 the  unique tensor of the nuclear $\cstar$-algebras $\AlL$ and $\AlR$, 
namely,  $\AlZ=\AlL \otimes \AlR$. 
It is well known that  there exists a unique product extension of  
  two arbitrary  $\ast$-automorphisms on 
disjoint (nuclear) $\cstar$-systems $\AlL$ and $\AlR$,  
 as a $\ast$-automorphism on $\AlZ$;  see II.9.6.1 of \cite{Blackadar}.

For  the fermion lattice  system $\AlZ$, 
the situation  becomes  complicated due to the grading structure as follows.
Take an arbitrary element  
$\sum_{k} A_k B_k \in \AlZ $, where each $A_k\in \AlL$ and $B_k\in \AlR$. 
Define
\begin{align} 
\label{eq:DEFtilal}
\tilalL\left(\sum_{k} A_k B_k\right):=  \sum_{k} \alL(A_k) B_k, \nonum\\ 
\tilalR\left(\sum_{k} A_k B_k\right):=  \sum_{k} A_k \alR(B_k).
\end{align} 
By the  defining formula,  $\tilalL$ and 
$\tilalR$ are linear  maps from $\AlZ$  onto $\AlZ$. 

We now verify  that the above $\tilalL$ and $\tilalR$ 
actually give well-defined  $\ast$-isomorphisms of $\AlZ$.
To this end,  take  arbitrary elements  $E, F \in \AlZ$.
In order to examine the effect of grading,
with no loss of generality,  we assume  the following  forms
\begin{align*} 
E=\sum_{k} A_{k} B_{k}\in \AlZ,\quad 
F&=\sum_{l} C_{l} D_{l}\in \AlZ, 
\end{align*} 
where 
\begin{align*} 
A_{k}, C_{l} \in \AlLe\ {\text{or}}\ \in \AlLo,\quad 
B_{k}, D_{l} \in \AlRe\ {\text{or}}\ \in \AlRo.
\end{align*} 
Due to the graded locality \eqref{eq:glocality}
\begin{align*}
EF&=  \left(\sum_{k} A_{k} B_{k} \right)  
\left(\sum_{l} C_{l} D_{l} \right)   
= \sum_{k}\sum_{l}  A_{k} B_{k} C_{l} D_{l}   
= \sum_{k}\sum_{l}  A_{k} (B_{k} C_{l}) D_{l}  \nonum\\
&= \sum_{k}\sum_{l}  A_{k} \bigl(\theta(B_{k},  C_{l})  C_{l} B_{k}\bigr) D_{l} = \sum_{k}\sum_{l} \theta(B_{k},  C_{l})  (A_{k}C_{l}) (B_{k} D_{l}),  
\end{align*}
where $\theta$ takes $\pm 1$
as  defined in \eqref{eq:theta}.
As $A_{k}C_{l}\in \AlL$  and $B_{k} D_{l}\in \AlR$,  
we compute 
\begin{align*}
\tilalL(EF)&=\sum_{k}\sum_{l} \theta(B_{k}, C_{l}) \alL(A_{k}C_{l}) B_{k} D_{l} =\sum_{k}\sum_{l}\theta(B_{k}, C_{l}) \alL(A_{k}) \bigl(\alL(C_{l}) B_{k}\bigr)  D_{l}
\nonum\\
&=
\sum_{k}\sum_{l}\theta(B_{k}, C_{l}) \alL(A_{k}) 
\Bigl( \theta(\alL(C_{l}), B_{k}) B_{k} \alL(C_{l})\Bigr)  D_{l}\nonum\\
&=
\sum_{k}\sum_{l}\theta(B_{k}, C_{l})
\theta(\alL(C_{l}), B_{k})  \bigl(\alL(A_{k})  B_{k}\bigr) \bigl(\alL(C_{l})  D_{l}\bigr).
\end{align*}
The term $\theta(B_{k}, C_{l})
\theta(\alL(C_{l}), B_{k})$ in  the final line 
 of the above is $1$, because $\alL$ preserves the grading $\THEL$, 
 the even-oddness of $\alL(C_{l})$ is same as that of $C_{l}$, and hence  
$\theta(\alL(C_{l}), B_{k})=\theta(C_{l}, B_{k})=\theta(B_{k}, C_{l})$. 
 Thus, 
\begin{align}
\label{eq:alLEF}
\tilalL(EF)&=
\sum_{k}\sum_{l}  \left(\alL(A_{k})  B_{k}\right) \left(\alL(C_{l})  D_{l}\right)\nonum\\
&=\left( \sum_{k}  \alL(A_{k}) B_{k} \right)
\left( \sum_{l} \alL (C_{l}) D_{l} \right)\nonum\\
&=\tilalL(E)
\tilalL(F)
\end{align}
Hence, we conclude that $\tilalL$ is a homomorphism of $\AlZ$.
Next,  we observe  that 
\begin{align*} 
E^{\ast}&=\Bigl( \sum_{k} A_{k} B_{k}\Bigr)^{\ast}
 = \sum_{k}  B_{k}^{\ast} A_{k}^{\ast} \nonum\\
&=\sum_{k} \theta(B_{k}^{\ast}, A_{k}^{\ast})  A_{k}^{\ast}  B_{k}^{\ast}
 =\sum_{k} \theta(A_{k}, B_{k})  A_{k}^{\ast}  B_{k}^{\ast},  
\end{align*} 
where we have used the fact  that the $\ast$-operation preserves the grading.
We compute   
\begin{align*}
\tilalL(E^{\ast})&=
\sum_{k} \theta(A_{k}, B_{k})  \alL(A_{k}^{\ast})  B_{k}^{\ast}
=\sum_{k} \theta(A_{k}, B_{k})  \alL(A_{k})^{\ast}  B_{k}^{\ast}
\nonum\\
&=\sum_{k} \theta(A_{k}, B_{k})  \left(B_{k}\alL(A_{k})\right)^{\ast} 
=\sum_{k} \theta(A_{k}, B_{k}) \overline{ \theta(B_{k}, \alL(A_{k}))}
 \left(    \alL(A_{k}) B_{k}\right)^{\ast} 
\nonum\\
&=\sum_{k} \theta(A_{k}, B_{k}) \theta(B_{k}, \alL(A_{k}))
 \left(    \alL(A_{k}) B_{k}\right)^{\ast}
=\sum_{k} \theta(A_{k}, B_{k}) \theta(B_{k}, A_{k})
 \left(    \alL(A_{k}) B_{k}\right)^{\ast}
\nonum\\
&=\sum_{k} 
 \left(  \alL(A_{k}) B_{k}\right)^{\ast}
= 
 \Bigl(\sum_{k}  \alL(A_{k}) B_{k}\Bigr)^{\ast}
={\tilalL(E)}^{\ast}
\end{align*}
Thus, $\tilalL$ preserves the  $\ast$-operation.
We now conclude  that  $\tilalL$ is a $\ast$-automorphism of $\AlZ$, 
 since it is surjective by definition.
Its inverse automorphism is  concretely given by 
\begin{align} 
\label{eq:GYAKU-L}
\tilalL^{-1}(\sum_{k} A_k B_k):=  \sum_{k} \alL^{-1}(A_k) B_k.
\end{align}  
In a completely analogous manner, 
we can see that $\tilalR$ is also a $\ast$-automorphism. 
 Its inverse  automorphism is concretely given by
\begin{align} 
\label{eq:GYAKU-R}
\tilalR^{-1}\left(\sum_{k} A_k B_k\right):=  \sum_{k} A_k \alR^{-1}(B_k).
\end{align}  

Now we  define the following automorphism 
\begin{equation} 
\label{eq:g-AUT-comp}
\alL \gotimes \alR:=\tilalL\circ \tilalR
(=\tilalR\circ \tilalL)  \in {\rm{Aut}}(\AlZ) 
\end{equation} 
as the composition of the  commuting
$\ast$-automorphisms $\tilalL$ and $\tilalR$ on $\AlZ$. 
By \eqref{eq:DEFtilal}, it satisfies 
 the desired product formula  \eqref{eq:g-AUT-action}.  
\end{proof}

\begin{remark}
\label{rem:FER-EXT-AUTO}
The fermion case in Proposition~\ref{prop:EXT-AUT}
may be regarded as 
 ``Joint extension of automorphisms  of subsystems for a CAR system," echoing the title of \cite{AM2003EXT}.
The crucial difference between here and \cite{AM2003EXT} is that 
 both automorphisms must be even, 
 whereas one of the prepared  states  can be non-even
 to construct their product extension. 
This stricter requirement for automorphisms can be understood as follows.
If either $\alL$ on $\AlL$ or $\alR$ on $\AlR$ 
 does not preserve the fermion grading, then its extension to $\AlZ$ as in \eqref{eq:DEFtilal} is invalid, and the product extension as in \eqref{eq:alpha-ext-FER} cannot exist. 
\end{remark}

The following proposition concerns the product extension of KMS states 
 prepared on disjoint regions.
The corresponding statement for  tensor product systems (such as the 
 quantum spin lattice system under consideration) 
is well known.  
In mathematical physics, it has been regarded as obvious, as seen for example 
 in \cite{PWKMS} and many others. 
 Hence, in the   proof below  we  focus on the fermion case only.
  
\begin{pro}[Product extension of KMS states]
\label{prop:EXT-KMS}
Let $\altL$ \textup{($t\in\R$)} be a $\cstar$-dynamics of $\AlL$, and let $\altR$ \textup{($t\in\R$)} be a $\cstar$-dynamics of $\AlR$.
Suppose that $\psiL$ is an $(\altL,\,\beta)$-KMS state on $\AlL$, 
 and that $\psiR$ is an $(\altR,\,\beta)$-KMS state on $\AlR$.
For the quantum spin lattice system, 
the product extension $\psiL\otimes \psiR$ of $\psiL$  and $\psiR$  
yields an $(\altL\otimes\altR ,\,\beta)$-KMS state on $\AlZ$.
For the fermion lattice system, 
assume that  each of $\altL$ and $\altR$ preserves the fermion grading on its respective system, that is,  
\begin{equation} 
\label{eq:altLR-PRESERVE-G}
\altL \circ \THEL= \THEL\circ  \altL
\ \text{and} \quad \altR\circ  \THER= \THER \circ \altR\ (t\in \R),
\end{equation}
and at least one of  $\psiL$ and $\psiR$ (possibly both) is even with respect to the fermion grading on its system,    
\begin{equation} 
\label{eq:psi-sorezore-EVEN}
\psiL \circ \THEL=\psiL \ {\text{or (possibly both)}} \  
\psiR \circ \THER=\psiR.
\end{equation}
Then, the product extension $\psiL \gotimes \psiR$ of $\psiL$ and $\psiR$ 
yields an $(\altL\gotimes\altR ,\,\beta)$-KMS state on $\AlZ$. 
\end{pro}

\begin{proof}
For  the quantum spin lattice system, 
this is already well known, see
  Proposition 13.1.12 of \cite{KAD2} and  Proposition 4.3 of  \cite{TAKE2book}.

For the fermion lattice system, owing to the evenness of both $\altL$ and $\altR$ \eqref{eq:altLR-PRESERVE-G}, 
a method  analogous to that used for tensor product systems applies, subject to some modifications to be detailed below.
First, by following  the  argument in Lemma 9.2.17 and Proposition 13.1.12 of 
  \cite{KAD2}, 
 it is enough   to verify  the KMS relation  only for 
 pairs of  monomial elements of the  form 
\begin{align*} 
E= A B\in \AlZ,\quad 
F= C D\in \AlZ, 
\end{align*} 
where 
\begin{align*} 
A, C \in \AlLe\ {\text{or}}\ \in \AlLo,\quad 
B, D \in \AlRe\ {\text{or}}\ \in \AlRo.
\end{align*} 

By the  KMS condition assumed on  the left-sided system $\AlL$, 
there exists a complex-valued function $F^{\LL}_{A, C}(z)$
 of $z\in \C$, which is continuous and bounded on the closed strip
  $0 \le \operatorname{Im} z \le \beta$, holomorphic
 on its interior, and satisfies    
\begin{equation}
\label{eq:KMS-L}
F^{\LL}_{A, C}(t)=\psiL \bigl(A \altL (C) \bigr),\quad   
F^{\LL}_{A, C}(t+i \beta)=\psiL \bigl(\altL(C)A \bigr), \quad (t\in \R).
\end{equation}
Similarly, by the 
  KMS condition assumed on  the right-sided system $\AlR$,
 there exists a complex-valued function $F^{\RR}_{B, D}(z)$
 of $z\in \C$, which is continuous and bounded on the closed strip
  $0 \le \operatorname{Im} z \le \beta$, holomorphic
 on its interior, and satisfies 
\begin{equation}
\label{eq:KMS-R}
F^{\RR}_{B, D}(t)=\psiR \bigl(B \altR (D) \bigr),\quad   
F^{\RR}_{B, D}(t+i \beta)=\psiR \bigl(\altR(D)B \bigr), \quad (t\in \R).
\end{equation}

 From  \eqref{eq:g-AUT-action} and \eqref{eq:altLR-PRESERVE-G}, 
by  some direct computation, we have 
\begin{align}
\label{eq:E-ALTF}
E  \altL\gotimes\altR (F)&=AB \altL \gotimes \altR (CD)\nonum \\
&=AB  \altL(C) \altR (D)
=
A\left(B \altL(C)\right) \altR (D)\nonum \\
&= A \left(\theta(B, \altL(C)) \altL(C)  B\right) \altR (D)\nonum \\
&=\theta(B, C)\left( A \altL(C)\right) \left( B \altR (D) \right), 
\end{align}
and
\begin{align}
\label{eq:ALTF-E}
 \altL\gotimes\altR (F) E&= \altL \gotimes \altR (CD) AB\nonum \\
&= \altL(C) \altR(D) AB
=\altR (C)
\left(\altR(D)A\right)B\nonum \\
&= \altR (C)  \left(\theta( \altR (D), A) A\altR(D)  \right) B\nonum \\
&= \theta(A, D)\left(\altL(C)A\right) \left(\altR(D)B\right).
\end{align}
Thus, from the product property of the fermionic product states \cite{AM2003EXT} and \eqref{eq:E-ALTF}, we have
\begin{align}
\label{eq:thetaBC}
\psiL\gotimes\psiR \bigl(E \altL\gotimes\altR (F) \bigr)
&=\theta(B, C)
\psiL \bigl(A \altL (C) \bigr)
\psiR \bigl(B \altR (D) \bigr),
\end{align}
and from \eqref{eq:ALTF-E},  
\begin{align}
\label{eq:thetaAD}
\psiL\gotimes \psiR \bigl( \altL\gotimes\altR (F) E\bigr)
&=
\theta(A, D)\psiL \left(  \altL(C)A\right) \psiR \left(   \altR (D)B \right).
\end{align}
By combining \eqref{eq:KMS-L}, \eqref{eq:KMS-R},
 \eqref{eq:thetaBC} and \eqref{eq:thetaAD}, we obtain
\begin{align}
\label{eq:LLRR-1}
\psiL\gotimes\psiR \bigl(E \altL\gotimes\altR(F) \bigr)
=\theta(B, C) F^{\LL}_{A, C}(t)F^{\RR}_{B, D}(t), 
\end{align}
and
\begin{align}
\label{eq:LLRR-2}
\psiL\gotimes\psiR \bigl( \altL\gotimes\altR (F) E\bigr)
=\theta(A, D)F_{A, C}^{\LL}(t+i \beta)F_{B, D}^{\RR}(t+i \beta).
\end{align}

We aim to relate \eqref{eq:LLRR-1}
 and \eqref{eq:LLRR-2} by the KMS condition by  
removing  the nuisance factors 
 $\theta(A, D)$ and $\theta (B, C)$.
This can be carried out   as follows.   
If $\theta(A, D)\ne \theta (B, C)$, then  
$\theta(A, D) \theta (B, C)=-1$, and 
the possible two cases are as follows: 

\bigskip
 
\noindent$\bullet$ Both $A$ and $D$ are odd, and either $B$ or $C$ (or both) is even,

or

\noindent$\bullet$ both $B$ and $C$ are odd, and either $A$ or $D$ (or both) is even.

\bigskip

\noindent In any such case, either $AC\in \AlL$  or $BD\in \AlR$, or both, 
 must  be  odd. Hence,  due to \eqref{eq:altLR-PRESERVE-G},   
 either $A \altL(C)$ (and $\altL(C)A$) or  $B\altR (D)$ (and $\altR(D)B$),  
 or both, must  be  odd.
 Accordingly, the expectation values  of  
\eqref{eq:LLRR-1} and \eqref{eq:LLRR-2} both 
vanish for all $t\in \R$.
Thus, it suffices to consider the following alternative cases: 
$\theta(A, D)= \theta (B, C)=1$, and 
$\theta(A, D)= \theta (B, C)=-1$. 
For the former case, set 
\begin{align}
\label{eq:FZEN-I}
F^{\LL, \RR}_{E, F}(z):=F^{\LL}_{A, C}(z)F^{\RR}_{B, D}(z), \quad z\in \C,  
\end{align}
and for  the latter case, set
\begin{align}
\label{eq:FZEN-minus}
F^{\LL, \RR}_{E, F}(z):=-F^{\LL}_{A, C}(z)F^{\RR}_{B, D}(z), \quad z\in \C.
\end{align}
Then the complex function $F^{\LL, \RR}_{E, F}(z)$ $(z\in \C)$ defined above 
 satisfies  the desired property.  
Namely, $F^{\LL, \RR}_{E, F}(z)$ $(z\in \C)$ is continuous and bounded on the closed strip
  $0 \le \operatorname{Im} z \le \beta$, holomorphic
 on its interior, and satisfies the KMS relation     
\begin{equation}
\label{eq:KMS-rel-EF}
F^{\LL, \RR}_{E, F}(t)=\psiL\gotimes\psiR \bigl(E \altL\gotimes\altR (F) \bigr),\quad F^{\LL, \RR}_{E, F}(t+i \beta)
=\psiL\gotimes\psiR \bigl( \altL\gotimes\altR (F) E\bigr). 
\end{equation}
Therefore, this completes the proof.
\end{proof}

\subsection{Donald's formula of quantum mutual entropy} 
\label{subsec:DONALD}
In this subsection, we introduce a notable identity of 
the quantum relative entropy, given in Equation (5.22) of \cite{OHYA-PETZ}. 
It is attributed  to Matthew J. Donald in \cite{OHYA-PETZ}, with no 
 original publication indicated.
Although this formula appeared in \cite{HO-SA} in the framework of algebraic quantum field theory, its usefulness, in particular for quantum statistical mechanics, does not seem to be well recognized.
We make essential use of  Donald's formula to derive a key estimate in  
the proof of Theorem~\ref{thm:ONE-FINITE}.
For this  purpose, we shall present it in the form of mutual entropy,  
both for quantum spin lattice systems and for fermion lattice systems. 

\begin{pro}[Donald's formula of quantum mutual entropy for 
 both quantum spin and fermion lattice systems]
\label{prop:DONALD}
Let $\varpi$ be any faithful state on $\AlZ$.
Let $\rhoL$ be any faithful state on $\AlL$, 
 and  $\rhoR$ be any faithful state on $\AlR$. 
Then, for the quantum spin lattice system, 
 the following identity concerning the quantum relative entropy 
holds, including the case  where  both sides are infinite{\rm{:}}
\begin{equation} 
\label{eq:Mutual-DON}
I_\varpi(\ZL:\ZR)
=S(\varpi \mid \rhoL \otimes \rhoR)-
S(\vrpL \mid \rhoL)-S(\vrpR \mid \rhoR).
\end{equation} 
For the fermion lattice system, assume in addition that 
 all  states $\varpi$ on $\AlZ$, 
 $\rhoL$ on $\AlL$, 
 and  $\rhoR$ on $\AlR$ are even.
Then the following identity also holds, including the case 
 where  both sides are infinite{\rm{:}}
\begin{equation} 
\label{eq:Mutual-DON-FER}
I_\varpi(\ZL:\ZR)
=S(\varpi \mid \rhoL \gotimes \rhoR)-
S(\vrpL \mid \rhoL)
-S(\vrpR \mid \rhoR).
\end{equation} 
\end{pro}

\begin{proof}
We note at the outset that  
$\rhoL$ and $\rhoR$ in the proposition need  not be marginal states of 
some state $\rho$ on $\AlZ$, although they may well be.

The proof for general tensor-product systems is given in Corollary~5.20 of \cite{OHYA-PETZ}. Since the quantum spin lattice system is a particular instance of a tensor-product system, with the algebraic structure $\AlZ=\AlL \otimes \AlR$, 
the identity \eqref{eq:Mutual-DON} follows immediately.

We now turn to the proof for the fermion lattice system. 
As shown in Corollary~5.20 of \cite{OHYA-PETZ}, the identity  
\begin{equation} 
\label{eq:CORpetz}
S(\varpi \mid \rhoL\gotimes\rhoR)=
S(\vrpL \mid \rhoL)+
S(\varpi \mid \vrpL\gotimes\rhoR)
\end{equation} 
follows from the following general relation (Theorem 5.15 of \cite{OHYA-PETZ}):
\begin{equation} 
\label{eq:THE515}
S(\varpi \mid \rhoL\circ E)=
S(\vrpL \mid \rhoL)+
S(\varpi \mid \varpi \circ E),
\end{equation} 
 where $E$ is now taken to be 
 the conditional expectation from $\AlZ$ onto $\AlL$, 
 relative to the product state $\vrpL\gotimes\rhoR$. 
 The unique existence of such a conditional expectation 
 follows from Theorem 4.7 of \cite{RMP-AM}, where the tracial state on $\AlR$ used there is to be replaced by $\rhoR$ on $\AlR$. 
Analogously,  we obtain 
\begin{equation} 
\label{eq:CORpetzSYM}
S(\varpi \mid \vrpL\gotimes\rhoR)
=S(\vrpR \mid \rhoR)+
S(\varpi \mid \vrpL\gotimes\vrpR).
\end{equation}
Since $S(\varpi \mid \vrpL\gotimes\vrpR)=I_\varpi(\ZL:\ZR)$, 
by combining \eqref{eq:CORpetz} and \eqref{eq:CORpetzSYM}, 
we obtain \eqref{eq:Mutual-DON-FER}.
\end{proof}

\subsection{Completion of the proof; the final step} 
\label{subsec:complete}
In this final subsection, we complete the proof of Theorem~\ref{thm:ONE-FINITE}
 using the results in Subsections~\ref{subsec:AG-LR}, \ref{subsec:PRO-EXT}, and \ref{subsec:DONALD}.

We  apply $\vp$ to 
$\varpi$, $\vppotL$ to $\rhoL$, and 
 $\vppotR$ to $\rhoR$ in Proposition~\ref{prop:DONALD}, 
as these are all KMS (modular) states. 
Accordingly, for the quantum spin lattice system, 
the formula  \eqref{eq:Mutual-DON} yields 
\begin{equation} 
\label{eq:DON-APPLY-KMS}
I_\vp(\ZL:\ZR)
=S(\vp \mid \vppotL \otimes \vppotR)-S(\vpL \mid \vppotL)-S(\vpR \mid \vppotR), \end{equation} 
and for the fermion lattice system, the formula \eqref{eq:Mutual-DON-FER}
 yields
\begin{equation} 
\label{eq:DON-APPLY-KMS-FER}
I_\vp(\ZL:\ZR)
=S(\vp \mid \vppotL\gotimes\vppotR)-S(\vpL \mid \vppotL)-S(\vpR \mid \vppotR).
\end{equation} 

By the positivity of relative entropy, 
$S(\vpL \mid \vppotL)\ge 0$ and $S(\vpR \mid \vppotR)\ge 0$, 
for the quantum spin lattice system, we have  
\begin{equation} 
\label{eq:ine-APPLY-KMS}
I_\vp(\ZL:\ZR)
\le S(\vp \mid \vppotL \otimes \vppotR), 
\end{equation} 
and for the fermion  lattice system, we have
\begin{equation} 
\label{eq:ine-APPLY-KMS-FER}
I_\vp(\ZL:\ZR)
\le S(\vp \mid \vppotL\gotimes\vppotR).
\end{equation} 

By  Proposition~\ref{prop:AG-PROD-LR}, 
for the quantum spin lattice system, 
\begin{equation} 
\label{eq:USIRO-PERs}
S(\vp \mid \vppotL \otimes \vppotR)=S(\vp \mid [\vpbW]),
\end{equation} 
and for the fermion lattice system, 
\begin{equation} 
\label{eq:USIRO-PERf}
S(\vp \mid \vppotL\gotimes\vppotR)=S(\vp \mid [\vpbW]).
\end{equation} 
By the formula for quantum relative entropy 
under perturbations in \cite{HP94} 
(see Remark~\ref{rem:HP} below for details), 
we obtain 
\begin{equation}
\label{eq:twoWLR}
S(\vp \mid [\vpbW])  \le 2 \beta \lVert \WLR \rVert.
\end{equation}

For the quantum spin system, the combination of 
\eqref{eq:ine-APPLY-KMS} \eqref{eq:USIRO-PERs} and  \eqref{eq:twoWLR},   
and for the fermion lattice system, the combination of \eqref{eq:ine-APPLY-KMS-FER} \eqref{eq:USIRO-PERf} and \eqref{eq:twoWLR}, respectively, yields 
the estimate 
\begin{equation} 
\label{eq:GOAL}
I_\vp(\ZL:\ZR)
\le 2 \beta \lVert \WLR \rVert.
\end{equation} 
This completes the proof of Theorem~\ref{thm:ONE-FINITE}.

\begin{remark}
\label{rem:HP}
For any self-adjoint element $h=h^{\ast}\in \Al$, consider 
 the perturbed state $[\omeh]$ of a modular state $\ome$ 
as  in Subsection \ref{subsec:GIBBS}.
From  the variational expression of the quantum relative entropy 
\cite{PETZ88}, which generalizes the 
 Golden-Thompson inequality  for von Neumann algebras  
 \cite{ARAKI-GTPB}, it follows that 
\begin{equation}
\label{eq:REL-perturbed-SA}
S\left([\omeh]  \mid \ome\right) \le \ome(h)-[\omeh](h) \le 2 \lVert h \rVert.
\end{equation}
By the chain rule property of the perturbed states  \cite{ARAKI-RH73}, 
 the inequality \eqref{eq:REL-perturbed-SA} also implies 
\begin{equation}
\label{eq:TSUKAU-REL-ENT-SA}
S\left(\ome  \mid [\omeh]\right)  \le 2 \lVert h \rVert.
\end{equation} 
\end{remark}

\begin{remark}
\label{rem:POT-ASSUMPTION}
As previously noted,  the assumption of 
 Theorem~\ref{thm:ONE-FINITE} is not optimal. 
For example, Dyson-type  one-dimensional classical lattice models  
with decay exponent $\alpha>2$ exhibit an analogous property  
as in Theorem~\ref{thm:ONE-FINITE}; see \cite{ENTER-new}.
\end{remark}

\section{Discussion}
\label{sec:DIS}
In this final section, we summarize our results  
and discuss some open problems from a broader perspective.

We have provided a mathematically rigorous definition of quantum  mutual entropy  in the quasi-local $\cstar$-system $\Al$ 
representing quantum spin lattice systems and fermion lattice systems 
in Section~\ref{sec:MUT-Def} under a fairly general setup.
 Our general formulation of the mutual entropy
 does not rely on  Tomita-Takesaki theory and  can also apply to ground (pure) states, as shown in Subsection~\ref{subsec:GR}.

With this quantum mutual entropy, 
we  have established the thermal area law   
 for quantum spin lattice 
 systems as in Theorem~\ref{thm:MAIN} 
and   for fermion  lattice systems as in Theorem~\ref{thm:FER-MAIN}. 

Our thermal area law in the $\cstar$-algebraic framework 
is derived from  the LTS condition rather than the KMS condition. 
For the potential $\pot$ in these theorems,  
 only the existence 
 of surface energies within the $\cstar$-system $\Al$ is assumed; 
a global time evolution on $\Al$ generated by $\pot$  is not required.
This generality is meaningful  both  physically and mathematically,
 since the \emph{surface} energy is the  essential ingredient for 
formulating the \emph{area} law and, from a technical perspective,  
 it is difficult to deduce the  
$\cstar$-dynamics from the mere existence of surface energies; such an existence has been established only for one-dimensional quantum lattice systems \cite{KIS-DD}.

\subsection{On extensions of LTS and the thermal area law}
\label{subsec:DISCUSSION-LTS}
The notion of LTS, as its  name suggests, is defined for 
  every local subsystem embedded in the infinitely extended  $\cstar$-system $\Al$. It is  suitable for the present purpose to treat local subsystems 
as open systems rather than  closed ones.  
Accordingly, the thermal area law holds 
for any \emph{specific} finite region $\I$, as shown in  
 Theorems~\ref{thm:MAIN} and \ref{thm:FER-MAIN}.
 This generality  suggests  a natural  extension of the thermal area law 
to \emph{metastable} states \cite{SEWrep}.

Conjecture 5.3.6 of \cite{SEWbook} proposes a generalization 
of the LTS condition from established quantum lattice systems to continuous quantum systems.  If such an extension is realized in certain 
 boson-field models on continuous spaces such as $\Rnu$ with $\nu \in \NN$, 
then a corresponding thermal area law  would follow, according to the model-independent proof of Theorem~\ref{thm:MAIN}. 
In particular, the operator-algebraic approach to the thermal area law 
for free boson models (see \cite{AW63} and \cite{BRA2}) may be worthwhile in comparison with the study  for free fermion models \cite{Bernigau}.

\subsection{Thermal destruction of quantum entanglement}
\label{subsec:Deduction}
The temperature dependence of quantum entanglement has been studied 
in several finite-qubit models \cite{ABV, NIE, TOTH}. The computations reported therein indicate that, in general, with some exceptions, quantum entanglement increases with the inverse temperature $\beta$ (or equivalently decreases with the temperature $T$).

In the present paper, we  consider certain infinite-qubit systems, namely,  
the quantum spin system and 
the fermion lattice system on $\Z$, with two subsystems 
(often  called Alice and Bob) given by the infinite left-sided and right-sided 
subsystems $\AlL$ and  $\AlR$.    
For critical ground states on $\Al$, which violate the split property,   
the quantum entanglement between $\AlL$ and $\AlR$ is infinite \cite{KMSW06}; 
this further enables ``embezzlement of entanglement"  \cite{CRITICAL}. 
Corollary~\ref{cor:KIERU} reveals a striking reduction of quantum entanglement from an infinite amount at $\beta=\infty$ to finite values for all $0\le \beta<\infty$.  Note that in \cite{HIGH} the disappearance of quantum entanglement at small $\beta$ (i.e., at high $T$) has been discussed, whereas  Corollary~\ref{cor:KIERU} concerns the behavior of quantum entanglement at and around $\beta=\infty$ (i.e., $T\approx0$). 

Taken together, these observations naturally raise the question of how the estimate given by  our thermal area law  can be affected by $\beta$, or possibly improved at certain values of $\beta$.

\subsection{Toward a thermal area law in algebraic quantum field theory}
\label{subsec:DIS-AQF}
The area law for vacuum states in AQFT has been formulated in \cite{HO-SA}. 
We raise the question of whether it is possible to formulate  a thermal area law in AQFT, in analogy with the case of quantum lattice systems discussed in this paper.  This problem appears to be intriguing, 
since thermal equilibrium (KMS) states of AQFT 
  can exhibit both 
 the split property \cite{DL-SPLIT} and the Reeh-Schlieder property \cite{RS}.
These two properties 
represent  somewhat contrasting aspects of state correlation---
independence versus quantum entanglement; see \cite{IWA}.

 The split property for KMS states with respect to a 
free quantum field model \cite{BUJU86}  
  has been investigated in \cite{Nozawa}, while 
the Reeh-Schlieder property for KMS states 
 has been shown under some general assumptions of AQFT in \cite{JAC}. 
We have derived the thermal area law from the LTS, 
a variational principle selecting thermal equilibrium states.
To our knowledge, a similar variational formulation  of relativistic KMS states has not yet been established. Such a direction may open up new possibilities for studying the temperature dependence of
 quantum entanglement in massive and massless quantum field models.

\subsection{Modular Hamiltonians of modular states}
For the closing of this paper, we suggest a new direction of research. 

The modular flows (modular automorphism groups) are a key concept  
in algebraic quantum field theory (AQFT); see \cite{IWA, HAAG}.
On local regions in Minkowski spacetime, 
a vacuum state gives rise to  modular states, and Tomita-Takesaki theory
enters as a crucial mathematical structure; 
we refer to  \cite{ARA-JMP64} as a pioneering work, and also \cite{FRE85}.
Another prominent example is the modular flow on Rindler wedges 
induced by the  vacuum state in Minkowski spacetime. 
It admits a clear geometric description as Lorentz boost transformations, forming the basis of the Unruh effect; see \cite{SEWPCT}.

Recently, the study of modular Hamiltonians (also called 
entanglement Hamiltonians) has flourished, with a wide range  of settings 
in both quantum field theory and quantum statistical mechanics; see e.g. \cite{CAS-HUE-LEC} and \cite{CARDY}.

In this paper, we  essentially consider modular Hamiltonians  
 of modular states. More precisely,   
a modular (KMS) state  on the infinitely extended total system $\Al$
 gives rise to modular states on local subsystems embedded in $\Al$ 
by restriction, whereas a vacuum state 
yields modular states on local subsystems in AQFT. 
Our setup falls into the class of  
`modular Hamiltonians for lattice models at finite temperature'  
 stated in Section 2.4 of \cite{DALMONTE-review}.

It is evident that the resulting modular Hamiltonians of KMS states 
 in quantum lattice systems are 
non-trivial unless the potential $\pot$ consists of
 one-point (i.e., non-interacting) interactions. 
Nonetheless, they still allow for control 
through the mutual entropy, as we have 
established in Theorems~\ref{thm:MAIN}, \ref{thm:FER-MAIN}, 
\ref{thm:ONE-FINITE}.

The discrepancy between the modular Hamiltonians (given by reduced states of a KMS state) and the local Hamiltonians (given directly by the potential $\pot$) has not been fully explored. The importance 
 of this discrepancy has been discussed in recent physics papers such as  
\cite{LocalTemp, MILLER}.  
However, within the $\cstar$-algebraic framework, the non-trivial nature of  
 this discrepancy had been addressed in several earlier works such as   
 \cite{ARA-ION}, \cite{HP94}, \cite{FELLER}, \cite{Master97}, and 
\cite{NAR}. The present paper may be regarded as one contribution
within this line of investigations, and we hope that it will stimulate further developments.

\bmhead{Acknowledgements}
This work was supported by Kakenhi (grant no. 21K03290) and Kanazawa University.

\section*{Declarations}

\begin{itemize}
\item Conflict of Interest Statement

The authors declare that they have  no conflict of interest. 

\item Data Availability Statement

No datasets were generated or analyzed during the current study.

\end{itemize}


\end{document}